\documentclass[11pt]{article}
\usepackage[colorlinks,citecolor=blue,bookmarks=true]{hyperref}
\usepackage{amsmath} 
\usepackage{amsthm} 
\usepackage{amssymb}	
\usepackage{graphicx} 
\usepackage{multirow}
\usepackage{color}
\usepackage[dvips,letterpaper,margin=1in]{geometry}
\usepackage[capitalize,noabbrev]{cleveref}
\usepackage{subcaption}

\usepackage[utf8]{inputenc}
\usepackage[english]{babel}

\usepackage{bm}
\usepackage{bbm}
\usepackage{diagbox}
\usepackage{mathtools}
\usepackage{ytableau}

\newtheorem{theorem}{Theorem}[section]

\newtheorem{lemma}[theorem]{Lemma}
\newtheorem{corollary}[theorem]{Corollary}
\newtheorem{proposition}[theorem]{Proposition}
\newtheorem{fact}[theorem]{Fact}

\newtheorem{definition}[theorem]{Definition}

\newtheorem{remark}[theorem]{Remark}

\newtheorem{question}[theorem]{Question}
\crefname{question}{Question}{Questions}

\DeclarePairedDelimiter\abs{\lvert}{\rvert}

\DeclarePairedDelimiter\ket{\lvert}{\rangle}
\DeclarePairedDelimiter\bra{\langle}{\rvert}

\newcommand{\E}{\mathop{\bf E\/}}

\newcommand{\avg}{\mathop{\textup{avg}\/}}

\newcommand{\tr} {\operatorname{tr}}
\newcommand{\poly} {\operatorname{poly}}

\newcommand{\supp} {\operatorname{supp}}

\newcommand{\Tab} {\textup{Tab}}

\newcommand{\Sh} {\textup{Sh}}

\newcommand{\ketbra}[2]{\ensuremath{\ket{#1}\!\bra{#2}}}

\newcommand{\kett}[1]{|#1\rangle\!\rangle}
\newcommand{\bbra}[1]{\langle\!\langle#1|}
\newcommand{\kettbbra}[2]{\ensuremath{\kett{#1}\!\bbra{#2}}}
\newcommand{\bbrakett}[2]{\ensuremath{\langle\!\langle{#1}\vert{#2}\rangle\!\rangle}}

\allowdisplaybreaks

\DeclarePairedDelimiter\parens{\lparen}{\rparen}
\DeclarePairedDelimiter\norm{\lVert}{\rVert}

\newcommand{\calA}{\mathcal{A}}
\newcommand{\calB}{\mathcal{B}}
\newcommand{\calC}{\mathcal{C}}

\newcommand{\calV}{\mathcal{V}}

\newcommand{\Davg}{\mathcal{D}_{\textup{avg}}}

\title{Approximation does not help in quantum unitary time-reversal}
\date{}

\author{
Kean Chen\thanks{University of Pennsylvania, Philadelphia, USA. Email: \texttt{keanchen.gan@gmail.com}}\and
Nengkun Yu\thanks{Stony Brook University, NY, USA. Email: \texttt{nengkunyu@gmail.com}}\and
Zhicheng Zhang\thanks{University of Technology Sydney, Sydney, Australia. Email: \texttt{iszczhang@gmail.com}}
}

\begin{document}

\maketitle

\begin{abstract}
Access to the time-reverse $U^{-1}$ of an unknown quantum unitary process $U$ is widely assumed in quantum learning, metrology, and many-body physics. 
The fundamental task of unitary time-reversal dictates implementing $U^{-1}$ to within diamond-norm error $\epsilon$ using black-box queries to the $d$-dimensional unitary $U$.
Although the query complexity of this task has been extensively studied, existing lower bounds either hold only for the exact case (i.e., $\epsilon=0$)
or are suboptimal in $d$.
This raises a central question: does approximation help reduce the query complexity of unitary time-reversal?
We settle this question in the negative by establishing a \textit{robust} and \textit{tight} lower bound $\Omega((1-\epsilon)d^2)$ with explicit dependence on the error $\epsilon$. This implies that unitary time-reversal retains optimal exponential hardness (in the number of qubits) even when constant error is allowed. Our bound applies to adaptive and coherent algorithms with unbounded ancillas and holds even when $\epsilon$ is an average-case distance error.
\end{abstract}

\section{Introduction}

The time evolution of a closed quantum system is governed by a reversible unitary $U$. 
Time-reversal of this evolution can, in principle, be simulated by implementing the inverse unitary $U^{-1}$.
However, this requires complete knowledge of the quantum system, which is typically inaccessible when the dynamics is dictated by nature.
A straightforward solution is to perform unitary tomography to obtain an approximate description of the unknown $d$-dimensional unitary $U$, which enables approximate unitary time-reversal with error $\epsilon$, using \(O(d^2/\epsilon)\) queries~\cite{haah2023query} to $U$.

In this paper, we study the \textit{unitary time-reversal} task: given black-box query access to an unknown $d$-dimensional unitary $U$, approximately implementing $U^{-1}$ to within diamond-norm error $\epsilon$.
Beyond its root as a natural concept in physics, this task is fundamental to quantum information theory and has a deep connection to the power of out-of-time-order correlators (OTOCs)~\cite{cotler2023information,schuster2023learning,shenker2014black,yao2016interferometric,swingle2016measuring,vermersch2019probing,xu2024scrambling} and advanced quantum learning algorithms with time-reversed unitary access~\cite{van2023quantum,wang2022quantum,gilyen2022improved,gs007,schuster2024random,zhao2025learning,tang2025amplitude}.
In quantum cryptography, there is also a growing attention~\cite{ma2025construct,Zhandry25,SML+25} to security against an adversary that can query time-reversed unitary oracles.

Quantum algorithms for unitary time-reversal have been extensively studied in the literature~\cite{sardharwalla2016universal,sedlak2019optimal,quintino2019probabilistic,quintino2019reversing,ebler2023optimal,quintino2022deterministic,navascues2018resetting,grinko2024linear,trillo2020translating,trillo2023universal,schiansky2023demonstration,yang2021representation,mo2025parameterized,zhu2024reversing,mo2025efficientinversionunknownunitary,yoshida2023reversing,chen2024quantum,grinko2025sequential,brzic2025higher,zhen2025structure}.
Surprisingly, Yoshida, Soeda, and Murao~\cite{yoshida2023reversing} demonstrated that, when the dimension $d=2$, the unitary time-reversal can be done \textit{deterministically} and \textit{exactly}\footnote{Exactness and determinism mean that the algorithm has no error and has success probability \(1\).} using only four queries to the unknown unitary. 
Thereafter, Chen, Mo, Liu, Zhang, and Wang~\cite{chen2024quantum} significantly generalized this result by providing a deterministic and exact time-reversal algorithm for unitaries of \textit{any dimension} \(d\), using \(O(d^2)\) queries.
In a complementary breakthrough, Odake, Yoshida, and Murao~\cite{odake2024analytical} showed that any deterministic and exact (i.e., $\epsilon=0$) unitary time-reversal algorithm must use at least $\Omega(d^2)$ queries, matching the upper bound in~\cite{chen2024quantum}.

However, the prior best lower bound $\Omega(d^2)$~\cite{odake2024analytical} holds only for the exact case.
A remaining gap is whether allowing approximation with a nonzero error $\epsilon> 0$ could yield any improvement over the $O(d^2)$ query complexity. 
In other words, the following question is still open:

\begin{question}\label{eq-1132218}\centering
\textit{Does approximation help reduce the query complexity of unitary time-reversal?
}
\end{question}

\subsection{Our results}

In this paper, we provide a negative answer to the above question by showing a \textit{robust} and \textit{tight} query lower bound $\Omega((1-\epsilon)d^2)$,
demonstrating that unitary time-reversal retains optimal exponential hardness (in the number of qubits)
even when constant error is allowed.
Our lower bound, combined with the matching upper bound $O(d^2)$ by the exact algorithm given in \cite{chen2024quantum}, fully settles the query complexity of the unitary time-reversal task with approximation.
Formally, our main result is as follows:
\begin{theorem}[\cref{thm-640254} and \cref{coro-6290408} restated]
\label{thm-6291832}
Given query access to an unknown $d$-dimensional unitary $U$,
any algorithm that approximates the time-reversed unitary $U^{-1}$ to within diamond norm or average-case distance\footnote{The average-case distance is given in \cref{def-6290002}.} error $\epsilon$, must use at least $\Omega((1-\epsilon)d^2)$ queries.
\end{theorem}

One interesting aspect of our result is the error scaling $(1-\epsilon)$, which fundamentally differs from the typical scaling $\poly(1/\epsilon)$ in quantum learning or metrology tasks.
This error scaling is natural as the unitary time-reversal can be done exactly (i.e., $\epsilon=0$).
Thus, \Cref{thm-6291832} provides a \textit{robust and optimal} hardness guarantee: permitting approximation even with constant errors cannot lead to better efficiency.
In comparison, most prior works focus on the lower bounds only for the \textit{exact} case,
which has recently been settled by Odake, Yoshida, and Murao~\cite{odake2024analytical} through a differentiation-based SDP framework (see further comparison in \cref{sec-710433,tab:main}).
We note that their method depends on the exact differentiation of the transformation map, and therefore cannot provide hardness guarantee in the approximate (non-exact) regime: it cannot rule out the possibility that an approximate unitary time-reversal algorithm could achieve better efficiency.

We note that our lower bound is strong in various senses.
First, it holds not only for diamond norm error $\epsilon$ (i.e., worst-case guarantee), but also for an average-case distance error $\epsilon$ (i.e., average-case guarantee over Haar random input states), a stronger robustness notion in many settings.
Next, it applies to the general adaptive and coherent algorithms with unbounded ancillas.
Using the techniques~\cite{Kitaev95,sheridan2009approximating,tang2025controlled} of unitary controlization,
it further extends to algorithms that have access to controlled queries.
Moreover, we note that our bound matches the numerical results obtained in \cite{grinko2025sequential} (see \cref{thm-290047} and the discussion thereafter), providing an additional empirical illustration of the asymptotic tightness established above, and further suggesting that our bound is likely optimal even non-asymptotically (i.e., without suppressing constant factors).

\paragraph{Lower bound for generalized time-reversal.}
We can provide a lower bound for the \textit{generalized unitary time-reversal} task: approximately implementing the time-reverse $U^{-t}\coloneqq e^{-iHt}$ for a given reversing time $t>0$.
Here, $H$ is allowed to be any Hamiltonian satisfying $e^{iH}=e^{i\theta}U$ for some real number $\theta$ and $\norm*{H}\leq \pi$, where $\|\cdot\|$ denotes the operator norm, the largest singular value.
Using \cref{thm-6291832}, we can provide a robust and tight query lower bound $\Omega(d^2)$ for this task with constant error and constant reversing time, which matches the upper bound $O(d^2)$ by unitary tomography~\cite{haah2023query}. 

\begin{corollary}[\cref{cor:generalized-time-reversal} restated]
\label{cor:main}
Given query access to an unknown $d$-dimensional unitary $U$, any algorithm that approximates the generalized time-reversed unitary $U^{-t}$ to within constant diamond norm error $\epsilon \leq 10^{-5}$ for constant reversing time $t\geq 0.1$, must use at least $\Omega(d^2)$ queries.
\end{corollary}

Here, the ranges $t\geq 0.1$ and $\epsilon\leq 10^{-5}$ are rather loose and chosen for simplicity of proof.

\paragraph{Lower bound for low-depth unitaries.}
Our main result, when combined with the low-depth construction of pseudorandom unitaries~\cite{ma2025construct,schuster2024random,SML+25}, can give a super-polynomial lower bound for the time complexity of unitary time-reversal even when the target unitary is restricted to depth $\poly(\log\log\log d)$\footnote{A $t$-depth unitary is a unitary that can be implemented by a $t$-depth circuit using single- and two-qubit gates}, where $\log d$ represents the number of qubits.
This strengthens the prior lower bound by \cite{schuster2024random} that requires larger depth $\poly(\log\log d)$.
Formally,
\begin{corollary}\label{coro-1141016}
Under cryptographic assumptions, no algorithm can solve the unitary time-reversal task in polynomial time (with respect to the number of qubits), even only for $\poly(\log\log\log d)$-depth unitaries and with constant average-case distance error.
\end{corollary}
\begin{proof}
The idea follows from noting that the Haar expectation in \cref{thm-691837} can be simulated using $\poly(\log\log\log d)$-depth pseudorandom unitaries~\cite[Corollary 2]{schuster2024random}. 
Specifically, we assume, to the contrary, that an algorithm $R$ performs unitary time-reversal for these low-depth pseudorandom unitaries in polynomial time. Then, the quantum channel shown in \cref{fig-6210108} must be close to the identity channel.
On the other hand, our \cref{thm-691837} shows that, for Haar random unitaries, the quantum channel shown in \cref{fig-6210108} must be close to the completely depolarizing channel.
This is because \cref{thm-691837} in fact shows that the maximal eigenvalue of the Choi representation of the channel in \cref{fig-6210108} is no more than $(n+1)/d$ and $n$ is polynomial in $\log d$ by assumption. 
This means the algorithm $R$ can be used to distinguish the pseudorandom and Haar random unitaries, by distinguishing whether the channel in \cref{fig-6210108} is close to identity or depolarizing channel to some constant precision, which only incurs a constant overhead (e.g., inputting $\ket{0}$ and then measuring $\{\ketbra{0}{0},I-\ketbra{0}{0}\}$). This contradicts the indistinguishability between pseudorandom unitaries and Haar random unitaries under cryptographic assumptions~\cite{ma2025construct}.
\end{proof}

Moreover, when combined with the low-depth construction of unitary $t$-designs~\cite[Theorem~1]{SML+25}, 
our result establishes a lower bound of $\Omega(t)$ for the unitary time-reversal of $\widetilde{O}(t)$\footnote{$\widetilde{\Omega}(\cdot)$ hides polylogarithmic factors.}-depth unitaries (for any $t\leq O(d^2)$) under the average-case distance. 
Recalling that an arbitrary $\log d$-qubit unitary can be implemented with depth $O(d^2)$~\cite{M_tt_nen_2004}, 
our lower bound $\Omega(t)$ is optimal up to polylogarithmic factors. 
This is because the range $t\leq O(d^2)$ of our lower bound saturates the upper bound of the depth of $\log d$-qubit unitaries, and general $O(d^2)$-depth unitaries can be learned with $\Theta(d^2)$ queries. 
This result can be viewed as a \emph{graded} version of our main lower bound. 
Notably, these graded bounds are purely information-theoretic and do not rely on the cryptographic assumptions 
required by pseudorandom unitaries.

\subsection{Overview of techniques}
Here, we briefly outline the main idea for the proof of \cref{thm-6291832}.
Suppose $R$ is a unitary time-reversal algorithm with error bounded by $\epsilon$ that uses $n$ queries to the unknown unitary $U$. Then, $R$ can be described by an $(n+1)$-comb~\cite{chiribella2008quantum,chiribella2009theoretical} (see \Cref{sec-6290132} for a brief introduction) on $(\mathcal{H}_0,\mathcal{H}_1,\ldots,\mathcal{H}_{2n+1})$ where each $\mathcal{H}_i$ is a copy of $\mathbb{C}^d$. 
We make an additional query to $U$ at the end of $R$, so that the overall channel approximates the identity channel
$\mathcal{I}_{\mathcal{H}_0\rightarrow \mathcal{H}_{2n+2}}$ from $\mathcal{H}_0$ to $\mathcal{H}_{2n+2}$ (see \cref{fig-6210108}).
Here, the $n+1$ queries to $U$ in total can be also described by an $(n+1)$-comb: 
\[C_U\coloneqq \kettbbra{U}{U}^{\otimes n+1},\]
where $\kett{U}$ denotes the vector obtained by flattening the matrix $U$.
Consequently, the overall channel is represented by the $1$-comb
\(
R \star C_U\stackrel{\epsilon}{\approx} \kettbbra{I_{\mathcal{H}_0\rightarrow\mathcal{H}_{2n+2}}}{I_{\mathcal{H}_0\rightarrow\mathcal{H}_{2n+2}}},
\)
where $\star$ denotes the link product~\cite{chiribella2009theoretical} for quantum combs, $\kettbbra{I_{\mathcal{H}_0\rightarrow\mathcal{H}_{2n+2}}}{I_{\mathcal{H}_0\rightarrow\mathcal{H}_{2n+2}}}$ is the Choi representation of the identity channel $\mathcal{I}_{\mathcal{H}_0\rightarrow\mathcal{H}_{2n+2}}$, and $\stackrel{\epsilon}{\approx}$ denotes approximation with error bounded by $\epsilon$. 
Taking Haar expectation over $U\in\mathbb{U}_d$ and using the bilinearity of $\star$, we can see that:
\begin{equation}\label{eq-721708}
R \star \E_U[C_U] \stackrel{\epsilon}{\approx} \kettbbra{I_{\mathcal{H}_0\rightarrow\mathcal{H}_{2n+2}}}{I_{\mathcal{H}_0\rightarrow\mathcal{H}_{2n+2}}}.
\end{equation}
Here, $\E_U[C_U]$ is called the $(n+1)$-th moment of the Haar random unitary under Choi representation (also called performance operator~\cite{quintino2022deterministic} for the time-reversal task).
Then, we use the \textit{stair operators} $\{A_k\}_{k=1}^n$ (see \cref{def-6100203}) as a tool to analyze the Haar moment $\E_U[C_U]$.
The stair operators are defined in the Young-Yamanouchi basis~\cite{ceccherini2010representation,harrow2005applications} under the Schur-Weyl duality of an asymmetric bipartite quantum system.
We prove two lemmas (see \cref{lemma-641548} and \cref{lemma-641551}) for the stair operators. 
The first lemma shows that they provide a good upper bound (w.r.t. the L\"owner order) for the Haar moment:
\begin{equation}\label{eq-721713}
\E_U[C_U]\sqsubseteq I_{\mathcal{H}_{2n+1}}\otimes A_n,
\end{equation}
where $I_{\mathcal{H}_{2n+1}}$ is the identity operator on $\mathcal{H}_{2n+1}$.
The second lemma indicates that, the stair operators are compatible with quantum combs through the link product $\star$, i.e., when linked with an arbitrary $(n+1)$-comb $R$, they can be sequentially contracted and bounded, until we obtain the identity operator (with a coefficient) on $\mathcal{H}_{2n+2}\otimes \mathcal{H}_0$, as shown in \cref{lemma-641443}:
\begin{equation}\label{eq-721714}
R\star (I_{\mathcal{H}_{2n+1}}\otimes A_n)\sqsubseteq\frac{n+1}{d} I_{\mathcal{H}_{2n+2}}\otimes I_{\mathcal{H}_0}.
\end{equation}
Combining \cref{eq-721713} and \cref{eq-721714}, we can show that
\begin{equation}\label{eq-750225}
R\star \E_U[C_U]\sqsubseteq R\star(I_{\mathcal{H}_{2n+1}}\otimes A_n)\sqsubseteq \frac{n+1}{d}I_{\mathcal{H}_{2n+2}}\otimes I_{\mathcal{H}_0},
\end{equation}
where we also use the property that $\star$ preserves the L\"owner order.
Intuitively, this means if $n$ is small and $U$ is Haar random, the overall channel $R\star \E_U[C_U]$ is highly depolarizing, for any algorithm $R$.
Therefore, combining \cref{eq-750225} with \cref{eq-721708}, we can see that $n$ must be $\Omega(d^2)$ even for a constant non-zero error $\epsilon$, since the maximum eigenvalue of $\kettbbra{I_{\mathcal{H}_0\rightarrow\mathcal{H}_{2n+2}}}{I_{\mathcal{H}_0\rightarrow\mathcal{H}_{2n+2}}}$ is $d$.
As our proof is based on the Choi representation, it is natural to expect that our lower bound applies even for the average-case distance error, according to \cref{eq-6282338}.

To prove these two lemmas, we perform a representation-theoretic analysis of the stair operators.
Specifically, we exploit the raising and lowering techniques of Young diagrams based on the branching rule of symmetric group representations to show \cref{eq-721713} and \cref{eq-721714}, respectively. 
Since the subsystem to be traced out is not the ``last'' subsystem (w.r.t. the ordering fixed by the action of symmetric groups) in the stair operators, the lowering techniques are not directly applicable here.
To address this, we use the Young's orthogonal form~\cite{ceccherini2010representation} for adjacent transpositions to swap the last two subsystems and then analyze the resulting expressions.
By this, we can reduce the original problems to pure combinatorics on Young diagrams (see \cref{sec-722129}), which is then solved with the assistance of Kerov's interlacing sequences~\cite{kerov1993transition,kerov2000anisotropic}.

\subsection{Related work}

In what follows, we review previous lower bound results for the unitary time-reversal task.

\subsubsection{Lower bounds from semidefinite programming}\label{sec-710433}
Quintino, Dong, Shimbo, Soeda, and Murao~\cite{quintino2019probabilistic,quintino2019reversing} developed a systematic semidefinite programming (SDP) approach to find probabilistic exact algorithms (i.e., exact conditioned on a flag output bit being ``success'') for unitary transformation tasks such as transposition, complex conjugation and time-reversal.
Query lower bounds are established for such algorithms;
in particular, $\Omega(d)$ queries are required for probabilistic exact unitary time-reversal. 
This approach is extended by Yoshida, Koizumi, Studzi{\'n}ski, Quintino, and Murao~\cite{YKS+24} to derive the same $\Omega(d)$ query lower bound but applicable to approximate unitary time-reversal.
Odake, Yoshida, and Murao~\cite{odake2024analytical} further developed a general framework for deriving lower bounds on deterministic exact unitary transformation tasks. Specifically, they consider transformations described by differentiable functions $f:\mathbb{U}_d\rightarrow \mathbb{U}_d$ and analyze the induced Lie algebra mappings obtained by differentiation at various points.
Within the Lie algebra framework, they formulate an SDP such that the number of queries required to implement $f(\cdot)$ is lower bounded by the SDP solution. Consequently, they show an $\Omega(d^2)$ query lower bound for deterministic exact unitary time-reversal (i.e., $f(U)=U^{-1}$), matching the upper bound $O(d^2)$ by the algorithm in~\cite{chen2024quantum}.

However, the differentiation-based method in \cite{odake2024analytical} depends on the exact form of the transformation map and appears challenging to extend directly to the analysis of approximate (non-exact) algorithms, so that their lower bound $\Omega(d^2)$ applies only to the \textit{exact} case. 
This leaves open the possibility that approximate unitary time-reversal algorithms may achieve better efficiency 
and motivates our central \Cref{eq-1132218}.

\begin{table}
\centering
\setlength{\tabcolsep}{1.5mm}{{\renewcommand{\arraystretch}{1.5}
\begin{tabular}{|cc|cl|}
\hline
\multicolumn{2}{|c|}{Algorithm Types}                                                                    & \multicolumn{2}{c|}{Lower Bounds}                                            \\ \hline
\multicolumn{1}{|c|}{\multirow{2}{*}{\shortstack{Exact\\ \\ ($\epsilon=0$)}}} & Deterministic & $\quad\Omega(d^2)\quad$      & \cite{odake2024analytical}                             \\ \cline{2-4} 
\multicolumn{1}{|c|}{}                                                                   & Probabilistic$^\dagger$ & $\quad\Omega(d)\quad$        & \cite{quintino2019probabilistic,quintino2019reversing} \\ \hline
\multicolumn{2}{|c|}{Ancilla-Free}                                                                       & $\quad\Omega(d^{1/4})\quad$  & \cite{cotler2023information}                           \\ \hline
\multicolumn{2}{|c|}{Clean$^\ddagger$}                                                                              & $\quad\Omega(d)\quad$        & \cite{gavorova2024topological}                         \\ \hline
\multicolumn{2}{|c|}{\multirow{3}{*}{General}}                                                           & $\quad\Omega(\sqrt{d})\quad$ & \cite{fefferman2015quantum}                            \\ \cline{3-4} 
\multicolumn{2}{|c|}{}                                                                                   & $\quad\Omega(d)\quad$        & \cite{YKS+24}                                          \\ \cline{3-4} 
\multicolumn{2}{|c|}{}                                                                                   & $\quad\Omega(d^2)\quad$      & Our~\Cref{thm-6291832}            \\ \hline
\end{tabular}
}}
\caption{Comparison of query lower bounds for unitary time-reversal. In the approximate (non-exact) settings, the lower bounds hold for constant error $\epsilon> 0$.
$^\dagger$Probabilistic exact algorithms refer to algorithms that are exact conditioned on a flag output bit being ``success''. $^\ddagger$Clean means the ancilla output must be returned in $\ket{0}$ assuming the algorithm is only unitary with postselection.}\label{tab:main}
\end{table}

\subsubsection{Lower bounds from topological obstructions}
Gavorov{\'a}, Seidel, and Touati~\cite{gavorova2024topological} established lower bounds for unitary transformation tasks from a topological perspective. 
By proving any $m$-homogeneous function from $\mathbb{U}_d$ to the $1$-sphere $\mathcal{S}^1$ must satisfy $d\mid m$ (i.e., $d$ divides $m$), they show that any \textit{clean} algorithm (i.e., the ancilla output must be returned in $\ket{0}$ assuming the algorithm is only unitary with postselection) for unitary time-reversal task requires $\Omega(d)$ queries. Note that the cleanness requirement is a strong constraint, since uncomputation of ancillas is generally hard when the algorithm has access only to the time-forward $U$.

\subsubsection{Lower bounds from quantum learning}
Cotler, Schuster, and Mohseni~\cite{cotler2023information} designed a quantum learning task such that when ancillas are not allowed: $\Omega(d^{1/4})$ queries to $U$ are required to solve it; and only $O(1)$ queries are sufficient given additional access to the time-reverse $U^{-1}$. This separation implies an $\Omega(d^{1/4})$ query lower bound for the unitary time-reversal task, but it applies only to algorithms without ancillas~\cite{schuster2024random}.
Schuster, Haferkamp, and Huang~\cite{schuster2024random} later established a stronger separation in time complexity without this limitation. Their new learning task leverages the low depth construction of pseudorandom unitaries~\cite{schuster2024random,ma2025construct} and the connectivity features of quantum dynamics~\cite{schuster2023learning}. Specifically, this task cannot be solved in polynomial time; but it can be solved in $O(1)$ time given additional access to the time-reverse $U^{-1}$. This implies a super-polynomial lower bound $\omega(\poly\log d)$ for the time complexity of unitary time-reversal.


\subsubsection{Lower bounds from quantum adversary method}
Fefferman and Kimmel~\cite{fefferman2015quantum} investigated the task of finding the pre-image of an unknown in-place permutation oracle.
Specifically, an in-place permutation oracle $O_\pi$ is a unitary that maps $\ket{i}$ to $\ket{\pi(i)}$, where $i\in \{1,\ldots,d\}$ and $\pi\in\mathfrak{S}_d$ is an unknown permutation on $d$ elements. 
The task is to find the index $i$ such that $\pi(i)=1$.
They proved that at least $\Omega(\sqrt{d})$ queries to the in-place permutation oracle $O_\pi$ are required to solve this task, using the techniques from the quantum adversary method~\cite{ambainis2000quantum}.
On the other hand, we can directly find the index if we are able to apply $(O_\pi)^{-1}$ on $\ket{1}$. 
Thus, their result implies an $\Omega(\sqrt{d})$ lower bound for the unitary time-reversal.

\subsubsection{Lower bounds from compressed oracle method.}
Tang and Wright~\cite{tang2025amplitude} showed a query lower bound $\Omega(\min\{d,1/\delta^2\})$ for amplitude estimation without access to $U^{-1}$,
revealing the importance of the time-reverse unitary in the Grover-type quadratic speedup.
Here, $d$ is the dimension of the unknown unitary $U$ and $\delta$ is the required precision.
Specifically, they reduced the problem of distinguishing two diagonal unitary ensembles to the amplitude estimation problem. The hardness of the former problem was then proved using the compressed oracle method~\cite{Zhandry19}, which purifies the randomness from the unitary ensembles and yields a finer lower bound given access only to $U$. 
Combined with the well-known upper bound $O(1/\delta)$, their result implies a lower bound of $\Omega(\min\{\sqrt{d},1/\sqrt{\epsilon}\})$ for unitary time-reversal to within diamond norm error $\epsilon$.\footnote{This can be seen by the following argument. Suppose $\epsilon\leq \delta^2$ and the time-reversal to within diamond norm error $\epsilon$ can be done using $n$ queries. By the original amplitude estimation algorithm, we can use $O(n/\delta)$ queries to $U$ to obtain an estimate of the amplitude with error $O(\delta+\epsilon/\delta)=O(\delta)$. Therefore, the lower bound due to \cite{tang2025amplitude} implies $n = \Omega(\min\{d\delta,1/\delta\})$. If $\epsilon\leq 1/d$, we set $\delta=1/\sqrt{d}$ and then $n=\Omega(\sqrt{d})$. Otherwise if $\epsilon>1/d$, we set $\delta=\sqrt{\epsilon}\geq 1/\sqrt{d}$ and then $n=\Omega(1/\sqrt{\epsilon})$.}
We note that when $\epsilon$ is a constant, this lower bound becomes $\Omega(1)$.

\subsection{Discussion and further implication}
In this paper, we settle the complexity of the unitary time-reversal, a fundamental task in quantum information processing.
We prove a robust and tight query lower bound $\Omega((1-\epsilon)d^2)$ for this task.
This answers the central Question~\ref{eq-1132218} in the negative: approximation does not help in the unitary time-reversal.
Here, we further discuss the implication of our results. 

An interesting question is how to develop a general framework for the robust hardness of quantum query transformations, like the SDP-based framework established by \cite{odake2024analytical} for the exact cases.
Beyond unitary time-reversal,
other two canonical query transformations are the unitary conjugation $U\mapsto U^{*}$ and unitary transposition $U\mapsto U^{\textup{T}}$.
For the unitary conjugation, it is shown that $O(d)$ queries suffice~\cite{miyazaki2019complex,ebler2023optimal} and $\Omega(d)$ queries are necessary even when allowing approximation~\cite{ebler2023optimal}.
For unitary transposition, our \Cref{thm-6291832} implies the \textit{state-of-the-art} result in the following \cref{coro-1140114}, extending the exact-case lower bound in \cite{odake2024analytical} to the approximate regime.
\begin{corollary}\label{coro-1140114}
Given query access to an unknown $d$-dimensional unitary $U$, any algorithm that approximates the transpose $U^{\textup{T}}$ to within diamond norm or average-case distance error $\epsilon$, must use at least $\Omega((1-\epsilon)d)$ queries.
\end{corollary}
This is by noting that any $m$-query unitary transposition algorithm combined with the $O(d)$-query exact unitary conjugation algorithm~\cite{miyazaki2019complex,ebler2023optimal}  yields an $O(md)$-query unitary time-reversal algorithm with the same approximation error.
However, there remains a gap to the best known upper bound $O(d^2)$ for unitary transposition by \cite{chen2024quantum}.

\subsection{Organization}

In \Cref{sec:pre}, we review the notation and preliminaries used in this paper. In \Cref{sec:hardness_of_time_reversal}, we present the details of our main results, where \Cref{thm-6291832} is proved through \Cref{sec-6290007,sec-641554,sec-641552,sec-641553}, and \Cref{cor:main} is proved in \Cref{sub:generalized_time_reversal}.
Technical lemmas and their proofs are deferred to \Cref{sec:deferred_lemmas}.

\section{Preliminaries}
\label{sec:pre}

\subsection{Notation}
We use $\mathcal{L}(\mathcal{H})$ to denote the set of linear operators on the Hilbert space $\mathcal{H}$. Similarly, we use $\mathcal{L}(\mathcal{H}_0,\mathcal{H}_1)$ to denote the set of linear operators from $\mathcal{H}_0$ to $\mathcal{H}_1$. Given two orthonormal bases for $\mathcal{H}_0$ and $\mathcal{H}_1$ respectively, we can represent each linear operator in $\mathcal{L}(\mathcal{H}_0,\mathcal{H}_1)$ by a $\dim(\mathcal{H}_1)\times \dim(\mathcal{H}_0)$ matrix and for such a matrix \(X\), we use \(\kett{X}\in \mathcal{H}_1\otimes\mathcal{H}_0\) to denote the vector obtained by flattening the matrix $X$. It is easy to see the following facts:
\[\kett{\ketbra{\psi}{\phi}}=\ket{\psi}\ket{\phi^*}, \quad\quad\quad \kett{XYZ}=X\otimes Z^\textup{T} \kett{Y},\]
where \(\ket{\phi^*}\) is the entry-wise complex conjugate of \(\ket{\phi}\) w.r.t. to a given orthonormal basis, and $Z^\textup{T}$ is the transpose of the matrix $Z$. 
The inner product can be denoted by \(\bbrakett{X}{Y}=\tr(X^\dag Y)\).

The \textit{Choi-Jamio{\l}kowski operator} of a quantum channel \(\mathcal{E}:\mathcal{L}(\mathcal{H}_0)\rightarrow \mathcal{L}(\mathcal{H}_1)\) is defined by:
\[\mathfrak{C}(\mathcal{E})=(\mathcal{E}\otimes\mathcal{I}_{\mathcal{H}_0})(\kettbbra{I}{I}),\] 
where \(\mathcal{I}_{\mathcal{H}_0}:\mathcal{L}(\mathcal{H}_0)\rightarrow\mathcal{L}(\mathcal{H}_0)\) is the identity channel and \(\kett{I}=\sum_i \ket{i}\ket{i}\) is the (unnormalized) maximally entangled state. Then, the application of the channel $\mathcal{E}$ can be described by its Choi-Jamio{\l}kowski operator:
\[\mathcal{E}(X)=\tr_{\mathcal{H}_0}\!\left(\mathfrak{C}(\mathcal{E})^{\textup{T}_{\mathcal{H}_0}}\cdot (I_{\mathcal{H}_1}\otimes X)\right),\]
where $\tr_{\mathcal{H}_0}$ and $\textup{T}_{\mathcal{H}_0}$ denote the partial trace and partial transpose on the subsystem $\mathcal{H}_0$, respectively.

For a sequence of Hilbert spaces $\mathcal{H}_0,\ldots,\mathcal{H}_{n}$ indexed by consecutive integers, we use $\mathcal{H}_{i:j}$ to denote the Hilbert space $\mathcal{H}_i\otimes\cdots\otimes\mathcal{H}_j$. For two linear operators $X,Y$, we use $X\sqsubseteq Y$ to denote that $Y-X$ is positive semidefinite.

\subsection{Distance between quantum channels}\label{sec:distance}

\paragraph{Worst-case distance.}
First, we introduce the diamond norm, which servers as the worst-case distance between quantum channels.
\begin{definition}[Diamond norm~\cite{AKN98}]
The \textit{diamond norm} of two quantum channels $\mathcal{E}_1,\mathcal{E}_2$ is defined as
\[\|\mathcal{E}_1-\mathcal{E}_2\|_\diamond\coloneqq \sup_{\rho}\|(\mathcal{E}_1\otimes \mathcal{I})(\rho)-(\mathcal{E}_2\otimes\mathcal{I})(\rho)\|_1,\]
where $\|\cdot\|_1$ denotes the Schatten $1$-norm (trace norm).
\end{definition}

There is a simple relation between the diamond norm of quantum channels and trace norm of the corresponding Choi-Jamio{\l}kowski operators. Specifically, let $\mathcal{E}_1,\mathcal{E}_2:\mathcal{L}(\mathbb{C}^d)\rightarrow\mathcal{L}(\mathbb{C}^d)$ be two $d$-dimensional quantum channels. Then, by the definition of diamond norm, we have
\begin{align}
\|\mathcal{E}_1-\mathcal{E}_2\|_\diamond 
&\geq \left\|(\mathcal{E}_1\otimes \mathcal{I})\!\left(\frac{\kettbbra{I}{I}}{d}\right)-(\mathcal{E}_2\otimes\mathcal{I})\!\left(\frac{\kettbbra{I}{I}}{d}\right)\right\|_1\nonumber\\
&=\frac{1}{d}\|\mathfrak{C}(\mathcal{E}_1)-\mathfrak{C}(\mathcal{E}_2)\|_1,\label{eq-692238}
\end{align}
where $\kettbbra{I}{I}/d$ is a maximally entangled state.

\paragraph{Average-case distance.}
Then, we introduce the average-case distance between quantum channels, which is widely used in quantum learning~\cite{huang2024learning,vasconcelos2024learning}, benchmarking~\cite{magesan2012characterizing,magesan2011scalable} and metrology~\cite{chen2023unitarity,yang2020optimal}. For this, we first need the notion of fidelity. Recall that the fidelity of two quantum states $\rho,\sigma$ is defined as 
\[F(\rho,\sigma)\coloneqq \tr\!\left(\sqrt{\sqrt{\sigma}\rho\sqrt{\sigma}}\right)^2.\] 
Note that when $\sigma=\ketbra{\psi}{\psi}$ is a pure state, we have $F(\rho,\sigma)=\tr(\rho\sigma)=\bra{\psi}\rho\ket{\psi}$.
Then, the definition of average-case distance is as follows:
\begin{definition}[Average-case distance~\cite{huang2024learning,vasconcelos2024learning}]\label{def-6290002}
The average-case distance between two quantum channels $\mathcal{E}_1,\mathcal{E}_2$ is defined as 
\[\mathcal{D}_{\textup{avg}}(\mathcal{E}_1,\mathcal{E}_2)\coloneqq \E_{\ket{\psi}}[1-F(\mathcal{E}_1(\ketbra{\psi}{\psi}),\mathcal{E}_2(\ketbra{\psi}{\psi}))],\]
where $\E_{\ket{\psi}}$ denotes the expectation over the Haar random state $\ket{\psi}$.
\end{definition}

A closely related measure is the \textit{channel fidelity}~\cite{raginsky2001fidelity} (or \textit{entanglement fidelity}~\cite{nielsen2002simple}). Specifically, for $d$-dimensional quantum channels $\mathcal{E}_1$ and $\mathcal{E}_2$, the channel fidelity $F_\textup{ch}$ of them is defined as the fidelity of their normalized Choi operators:
\[F_\textup{ch}(\mathcal{E}_1,\mathcal{E}_2)\coloneqq F\!\left(\frac{\mathfrak{C}(\mathcal{E}_1)}{d},\frac{\mathfrak{C}(\mathcal{E}_2)}{d}\right).\]
This measure is also commonly used in higher-order quantum transformation literature (see, e.g., 
\cite{TMM+25}
).

Now, we introduce some properties of these measures.
First, it is easy to see that both the average-case distance and the channel fidelity are unitarily invariant, i.e.,
\begin{equation}\label{eq-6292334}
\begin{split}
\Davg(\mathcal{U}\circ\mathcal{E}_1\circ\mathcal{V},\,\mathcal{U}\circ\mathcal{E}_2\circ\mathcal{V})&=\Davg(\mathcal{E}_1,\mathcal{E}_2),\\
F_{\textup{ch}}(\mathcal{U}\circ\mathcal{E}_1\circ\mathcal{V},\,\mathcal{U}\circ\mathcal{E}_2\circ\mathcal{V})&=F_{\textup{ch}}(\mathcal{E}_1,\mathcal{E}_2),
\end{split}
\end{equation}
for any unitary channels $\mathcal{U},\mathcal{V}$.
Second, the average-case distance is closely related to the channel fidelity when one of $\mathcal{E}_1$ and $\mathcal{E}_2$ is a unitary channel~\cite{nielsen2002simple}:
\begin{equation}\label{eq-6282338}
\Davg(\mathcal{E},\mathcal{U})= \frac{d}{d+1}\left(1-F_\textup{ch}\!\left(\mathcal{E},\mathcal{U}\right)\right)=\frac{d}{d+1}\left(1-\frac{1}{d^2}\tr(\mathfrak{C}(\mathcal{E})\cdot\kettbbra{U}{U})\right),
\end{equation}
where $d$ is the dimension of $\mathcal{U}$ and $\mathcal{E}$.
Then, it is easy to see that
\begin{align}
\Davg(\mathcal{E},\mathcal{U})&\leq \frac{d}{d+1}\left\|\frac{\mathfrak{C}(\mathcal{E})}{d}-\frac{\kettbbra{U}{U}}{d}\right\|_1
\leq \frac{d}{d+1}\|\mathcal{E}-\mathcal{U}\|_\diamond,\label{eq-6282353}
\end{align}
where the first inequality is due to the well-known property~\cite{fuchs2002cryptographic}: $1-\sqrt{F(\rho,\sigma)}\leq \frac{1}{2}\|\rho-\sigma\|_1$ and $F(\rho,\sigma)\leq 1$, the second inequality is by \cref{eq-692238}. 

We remark that the average-case distance can be much smaller than the diamond norm. For example, let $U=(\ketbra{0}{1}+\ketbra{1}{0})\oplus I_{d-2}$ be a $d$-dimensional unitary, where $I_{d-2}$ is the identity operator of dimension $d-2$. In such case, we have
\(\Davg(\mathcal{U},\mathcal{I})= O(1/d)\) but \(\|\mathcal{U}-\mathcal{I}\|_\diamond=\Omega(1)\). This implies that hardness results with average-case distance are generally much stronger than those with diamond norm. 

\subsection{Quantum combs}\label{sec-6290132}
In this section, we introduce the quantum comb~\cite{chiribella2008quantum,chiribella2009theoretical}, which is a powerful tool to describe (higher) transformations of quantum processes. Specifically, the Choi-Jamio{\l}kowski representation of quantum channels (i.e., transformations of quantum states) can be generalized to a higher-level concept (i.e., transformations of quantum processes), which is called \textit{quantum comb}.
\begin{definition}[Quantum comb~\cite{chiribella2009theoretical}]\label{def-681501}
For an integer $n\geq 1$, a quantum $n$-comb defined on a sequence of $2n$ Hilbert spaces $(\mathcal{H}_0,\mathcal{H}_1,\ldots,\mathcal{H}_{2n-1})$ is a positive semidefinite operator $X$ on $\bigotimes_{j=0}^{2n-1} \mathcal{H}_{j}$ such that there exists a sequence of operators $X^{(n)}, X^{(n-1)},\ldots, X^{(1)}, X^{(0)}$ such that
\begin{equation}\label{eq-681546}
\begin{split}
\tr_{\mathcal{H}_{2j-1}}\!\left(X^{(j)}\right)&=I_{\mathcal{H}_{2j-2}}\otimes X^{(j-1)},\quad 1\leq j \leq n,  
\end{split}
\end{equation}
where $X^{(n)}=X$ and $X^{(0)}=1$.
\end{definition}
Note that in \cref{def-681501}, the operators $X^{(j)}$ are again quantum $j$-combs (for convenience, we define any quantum $0$-comb to be the scalar $1$).
It is also worth noting that a quantum $1$-comb is simply the Choi-Jamio{\l}kowski operator of a quantum channel.
\begin{figure}[t]
    \centering
    \includegraphics[width=0.6\linewidth]{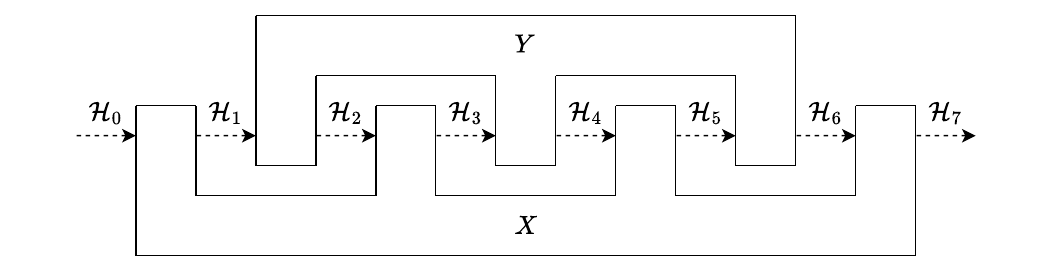}
    \caption{The combination of a $4$-comb $X$ with a $3$-comb $Y$, resulting in a $1$-comb $X\star Y$ on $(\mathcal{H}_0,\mathcal{H}_7)$.}
    \label{fig-6210015}
\end{figure}

The combination of any two quantum combs is defined by the link product ``$\star$'': 
\begin{definition}[Link product ``$\star$''~\cite{chiribella2009theoretical}]
\label{def-720255}
Suppose $X$ is a linear operator on $\mathcal{H}_{\bm{i}}=\mathcal{H}_{i_1}\otimes\mathcal{H}_{i_2}\otimes\cdots\otimes\mathcal{H}_{i_{n}}$ and $Y$ is a linear operator on $\mathcal{H}_{\bm{j}}=\mathcal{H}_{j_1}\otimes\mathcal{H}_{j_2}\otimes\cdots\otimes\mathcal{H}_{j_{m}}$,
where $\bm{i}=(i_1,\ldots,i_n)$ is a sequence of pairwise distinct indices, and similar for $\bm{j}=(j_1,\ldots,j_m)$.
Let $\bm{a}=\bm{i}\cap\bm{j}$ be the set of indices in both $\bm{i}$ and $\bm{j}$ and $\bm{b}=\bm{i}\cup\bm{j}$ be the set of indices in either $\bm{i}$ or $\bm{j}$.
Then, the combination of $X$ and $Y$ is defined by
\[X\star Y= \tr_{\mathcal{H}_{\bm{a}}}\!\left(X^{\textup{T}_{\mathcal{H}_{\bm{a}}}} \cdot Y\right)=\tr_{\mathcal{H}_{\bm{a}}}\!\left(X\cdot Y^{\textup{T}_{\mathcal{H}_{{\bm{a}}}}}\right),\]
where $\mathcal{H}_{\bm{a}}$ means the tensor product of subsystems labeled by the indices in $\bm{a}$, $\textup{T}_{\mathcal{H}_{\bm{a}}}$ means the partial transpose on $\mathcal{H}_{\bm{a}}$, both $X$ and $Y$ are treated as linear operators on $\mathcal{H}_{\bm{b}}$, extended by tensoring with the identity operator as needed.
\end{definition}

The link product has a very good property~\cite[Theorem 2]{chiribella2009theoretical}:
\begin{equation}\label{eq-721843}
X,Y\sqsupseteq 0 \Longrightarrow X\star Y\sqsupseteq 0.
\end{equation}
Moreover, it characterizes the channel concatenation under the Choi representation: given two quantum channels $\mathcal{E}_1:\mathcal{L}(\mathcal{H}_1)\rightarrow\mathcal{L}(\mathcal{H}_2)$ and $\mathcal{E}_2:\mathcal{L}(\mathcal{H}_2)\rightarrow\mathcal{L}(\mathcal{H}_3)$, we have 
$\mathfrak{C}(\mathcal{E}_2\circ \mathcal{E}_1)=\mathfrak{C}(\mathcal{E}_2)\star\mathfrak{C}(\mathcal{E}_1).$
The link product is also used to describe the combination of two higher-order combs such as that shown in \cref{fig-6210015}.
More generally, suppose $X$ is an $n$-comb on $(\mathcal{H}_0,\mathcal{H}_1,\ldots,\mathcal{H}_{2n-1})$ and $Y$ is an $(n-1)$-comb on $(\mathcal{H}_1,\mathcal{H}_2\,\ldots,\mathcal{H}_{2n-2})$, then
\begin{align}
X\star Y= \tr_{\mathcal{H}_{1:2n-2}}\!\left(X^{\textup{T}_{\mathcal{H}_{1:2n-2}}}\cdot (I_{\mathcal{H}_{2n-1}}\otimes Y \otimes I_{\mathcal{H}_0})\right)=\tr_{\mathcal{H}_{1:2n-2}}\!\left(X\cdot (I_{\mathcal{H}_{2n-1}}\otimes Y^{\textup{T}} \otimes I_{\mathcal{H}_0})\right)\nonumber
\end{align}
turns out to be a $1$-comb on $(\mathcal{H}_0,\mathcal{H}_{2n-1})$.
To see this, first note that $X,Y$ are positive semidefinite, which means $X\star Y$ is also positive semidefinite. On the other hand, note that
\begin{align}
\tr_{\mathcal{H}_{2n-1}}(X\star Y)&=\tr_{\mathcal{H}_{1:2n-2}}\!\left(\tr_{\mathcal{H}_{2n-1}}(X)^{\textup{T}_{\mathcal{H}_{1:2n-2}}}\cdot (Y \otimes I_{\mathcal{H}_0})\right)\nonumber\\
&=\tr_{\mathcal{H}_{1:2n-2}}\!\left((I_{\mathcal{H}_{2n-2}}\otimes X^{(n-1)})^{\textup{T}_{\mathcal{H}_{1:2n-2}}}\cdot (Y\otimes I_{\mathcal{H}_0})\right)\label{eq-681459}\\
&=\tr_{\mathcal{H}_{1:2n-3}}\!\left((X^{(n-1)})^{\textup{T}_{\mathcal{H}_{1:2n-3}}}\cdot (\tr_{\mathcal{H}_{2n-2}}(Y)\otimes I_{\mathcal{H}_0})\right)\nonumber \\
&=\tr_{\mathcal{H}_{1:2n-3}}\!\left((X^{(n-1)})^{\textup{T}_{\mathcal{H}_{1:2n-3}}}\cdot (I_{\mathcal{H}_{2n-3}}\otimes Y^{(n-2)}\otimes I_{\mathcal{H}_0})\right)\label{eq-681500}\\
&=\tr_{\mathcal{H}_{1:2n-3}}\!\left(((X^{(n-1)})^{\textup{T}_{\mathcal{H}_{1:2n-4}}})^{\textup{T}_{\mathcal{H}_{2n-3}}}\cdot (I_{\mathcal{H}_{2n-3}}\otimes Y^{(n-2)}\otimes I_{\mathcal{H}_0})\right)\nonumber\\
&=\tr_{\mathcal{H}_{1:2n-3}}\!\left((X^{(n-1)})^{\textup{T}_{\mathcal{H}_{1:2n-4}}}\cdot (I_{\mathcal{H}_{2n-3}}\otimes Y^{(n-2)}\otimes I_{\mathcal{H}_0})\right)\label{eq-681503}\\
&=\tr_{\mathcal{H}_{2n-3}}\!\left(X^{(n-1)}\star Y^{(n-2)}\right),\nonumber
\end{align}
where in \cref{eq-681459}, $X^{(n-1)}$ is an $(n-1)$-comb obtained from taking partial trace on $X$ (see \cref{def-681501}) and similarly in \cref{eq-681500}, $Y^{(n-2)}$ is an $(n-2)$-comb obtained from $Y$, \cref{eq-681503} is due to the property of partial trace $\tr_{\mathcal{H}}(X^{\textup{T}_{\mathcal{H}}})=\tr_{\mathcal{H}}(X)$. Then, we can iterate this process and finally obtain
\[\tr_{\mathcal{H}_{2n-1}}(X\star Y)=\tr_{\mathcal{H}_1}\!\left(X^{(1)}\right)=I_{\mathcal{H}_0},\]
which means $X\star Y$ is a $1$-comb on $(\mathcal{H}_0,\mathcal{H}_{2n-1})$.

\subsection{Young diagrams}\label{sec-6150256}
\paragraph{Young diagrams and standard Young tableaux.} A \textit{Young diagram} \(\lambda\) consisting of \(n\) boxes and \(\ell\) rows is a partition \((\lambda_1,\ldots,\lambda_\ell)\) of \(n\) such that \(\sum_{i=1}^\ell \lambda_i=n\) and \(\lambda_1\geq\cdots \geq \lambda_\ell> 0\). The $i$-th rows consists of $\lambda_i$ boxes. By convention, the Young diagram is drawn with left-justified rows, arranged from top to bottom.
The \textit{conjugate of a Young diagram} is obtained by exchanging rows and columns. For example, the Young diagram \(\lambda=(4,3,1)\) and its conjugate $\lambda'=(3,2,2,1)$ are drawn as:
\[\vcenter{\hbox{\scalebox{0.9}{\begin{ytableau}~ &~&~&~\\~&~&~\\~\end{ytableau}}\quad\quad\textup{ and}\quad\quad \scalebox{0.9}{\begin{ytableau}~&~&~\\~&~\\~&~\\~\end{ytableau}}}},\]
respectively.
We use $\ell(\lambda)$ to denote the number of rows of $\lambda$.
We use \(\lambda\vdash n\) to denote that \(\lambda\) is a Young diagram with \(n\) boxes and use $\lambda\vdash_d n$ to denote that $\lambda\vdash n$ and $\ell(\lambda)\leq d$.

A \textit{standard Young tableau} of shape \(\lambda\vdash n\) is obtained by a filling of the \(n\) boxes of \(\lambda\) using the integers \(1,\ldots,n\), each appearing exactly once, such that the entries increase from left to right in each row and from top to bottom in each column. An example of standard Young tableau of shape $(4,3,1)$ is
\[\vcenter{\hbox{\scalebox{0.9}{\begin{ytableau}1&3&4&6\\2&7&8\\5\end{ytableau}}}}.\]
Without causing confusion, we refer to the box in $T$ containing the integer $i$ as the $i$-box.
We use $\Sh(T)$ to denote the shape (Young diagram) of a standard Young tableau $T$.
We use $\Tab(\lambda)$ to denote the set of all standard Young tableaux $T$ such that $\Sh(T)=\lambda$ and use $\Tab(n)$ to denote $\bigcup_{\lambda\vdash n}\Tab(\lambda)$. 
Given a Young diagram $\lambda\vdash n$, the number of standard Young tableaux of shape $\lambda$ (i.e., $|\Tab(\lambda)|$) is characterized by the \textit{hook length formula}~\cite{fulton1997young}: 
\begin{equation}\label{eq-6192141}
|\Tab(\lambda)|=\frac{n!}{\prod_{\square\in\lambda} h_\lambda(\square)},
\end{equation}
where $h_\lambda(\square)$ is the number of boxes that are directly to the right and below the box $\square$, including the box itself (i.e., the ``hook length'').

\paragraph{Adjacency relation.} For $\lambda\vdash n+1$, $\mu\vdash n$, we write $\mu\nearrow\lambda$ or $\lambda\searrow\mu$ if $\lambda$ can be obtained from $\mu$ by adding one box. For a standard Young tableau $T$, we use $T^{\downarrow}$ to denote the standard Young tableau obtained from $T$ by removing the box containing the largest integer in $T$. 
Suppose $\mu\nearrow\lambda$ and $T\in\Tab(\mu)$. We use $T^{\uparrow\lambda}$ to denote the standard Young tableau obtained from $T$ by adding a box containing number $n+1$, resulting in shape $\lambda$.
We define $\Tab(\lambda,\mu)\subseteq \Tab(\lambda)$ to be the set of all standard Young tableaux $T$ such that $\Sh(T)=\lambda$ and $\Sh(T^\downarrow)=\mu$. It is easy to see that:
\begin{equation*}
|\Tab(\lambda,\mu)|=|\Tab(\mu)|.
\end{equation*}

Let $\mu,\nu\vdash n$ and $\mu\neq\nu$. We call $\mu$ and $\nu$ are \textit{adjacent} to each other if one of the following equivalent conditions is true:
\begin{itemize}
\item There exists a $\lambda\vdash n+1$ such that $\mu\nearrow\lambda$ and $\nu\nearrow\lambda$.
\item There exists a $\tau\vdash n-1$ such that $\tau\nearrow\mu$ and $\tau\nearrow\nu$.
\end{itemize}
Such $\mu$ and $\nu$ can be obtained from each other by moving a removable box. Furthermore, $\lambda=\mu\cup\nu$ and $\tau=\mu\cap\nu$ are unique, where $\mu\cup\nu$ denotes the Young diagram consisting of the union of boxes in $\mu$ and in $\nu$, and $\mu\cap\nu$ denotes the Young diagram consisting of the intersection of boxes in $\mu$ and in $\nu$. Moreover, we use $\lambda\setminus\mu$ to denote the set of boxes (not a Young diagram) in $\lambda$ but not in $\mu$.

\paragraph{Kerov's interlacing sequences.}
Kerov's interlacing sequences~\cite{kerov1993transition,kerov2000anisotropic} serve as an important tool in studying the combinatorics of Young diagram and their analytic generalizations.
Suppose $\lambda\vdash n$ is a Young diagram. 
The rotated version of $\lambda$ is defined by rotating the Young diagram $\lambda'$ (the conjugate of $\lambda$) $135^\circ$ counterclockwise. 
We define the $x$-axis so that the center of each box lies at an integer $x$-axis coordinate. Consider the shape of the rotated Young diagram extended with the linear curves $y=x$ and $y=-x$ (e.g., the bold line in \cref{fig-6152215}), which is a piecewise linear function. Then, Kerov's interlacing sequences are defined as the sequence $x$-coordinates of the local minima $\bm{\alpha}=(\alpha_1,\alpha_2,\ldots,\alpha_{L})$ and the sequence of $x$-coordinates of the local maxima $\bm{\beta}=(\beta_1,\beta_2,\ldots,\beta_{L-1})$. The two sequences interlace
\[\alpha_1< \beta_1<\alpha_2<\cdots<\alpha_{L-1}<\beta_{L-1}<\alpha_L.\]
\cref{fig-6152215} shows an example of the Young diagram $(6,5,3,3,2,1,1,1)$ and its interlacing sequences are $\bm{\alpha}=(-8,-4,-2,1,4,6)$ and $\bm{\beta}=(-7,-3,-1,3,5)$.
\begin{figure}[t]
    \centering
    \begin{subfigure}[b]{0.45\linewidth}
    \centering
    \includegraphics[width=1.0\linewidth]{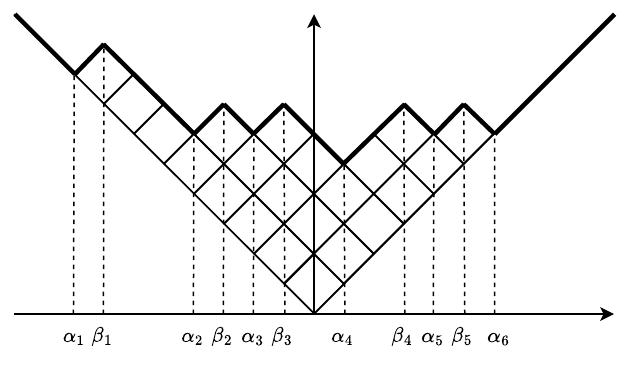}
    \vspace{-6mm}
    \caption{The interlacing sequences of $\lambda$.}\label{fig-6152215}
    \end{subfigure}
    \hspace{-0mm}
    \begin{subfigure}[b]{0.45\linewidth}
    \centering
    \includegraphics[width=1.0\linewidth]{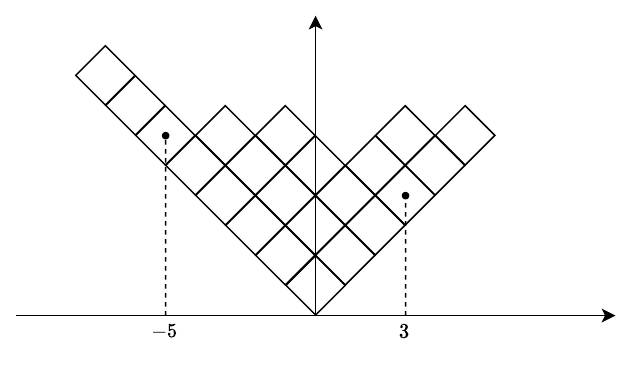}
    \vspace{-6mm}
    \caption{Axial coordinates of two boxes in $\lambda$.}\label{fig-6152350}
    \end{subfigure}
    \vspace{-0mm}
    \caption{Rotated version of the Young diagram $\lambda=(6,5,3,3,2,1,1,1)$.}\label{fig-6190458}
\end{figure}

We define the \textit{axial coordinate}\footnote{The axial coordinate is also called \textit{content} in the literature~\cite{ceccherini2010representation,okounkov1996new}} $c(\square)$ for a box $\square$ in a Young diagram as the $x$-axis coordinate of the center of this box when drawing in the rotated version. 
\cref{fig-6152350} shows an example of the Young diagram $(6,5,3,3,2,1,1,1)$ and the axial coordinates of two boxes. The \textit{axial distance} from box $\square_1$ to box $\square_2$ is defined as the difference $c(\square_1)-c(\square_2)$.

Adding or removing a box from a Young diagram can be described naturally using Kerov's interlacing sequences. A position where a box can be added corresponds precisely to a local minima $\alpha_i$. Conversely, a removable box corresponds to a local maxima $\beta_i$. Specifically, the axial coordinate of the added (or removed) box $\square$ (i.e., $c(\square)$) coincides with the $\alpha_i$ (or $\beta_i$). Suppose $\mu\vdash n$ has interlacing sequences $\alpha_1,\ldots,\alpha_L$ and $\beta_1,\ldots,\beta_{L-1}$. Let $\lambda\vdash n+1$ be obtained from $\mu$ by adding a box at $\alpha_k$. Then, we have~\cite[Eq. (3.2.3)]{kerov1993transition}
\begin{equation}\label{eq-6192042}
\frac{|\Tab(\lambda)|}{|\Tab(\mu)|}=(n+1)\prod_{i=1}^{k-1}\frac{\alpha_k-\beta_i}{\alpha_k-\alpha_i}\prod_{i=k+1}^L\frac{\alpha_k-\beta_{i-1}}{\alpha_k-\alpha_i}.
\end{equation}
On the other hand, suppose $\lambda\vdash n$ has interlacing sequences $\alpha_1,\ldots,\alpha_L$ and $\beta_1,\ldots,\beta_{L-1}$. Let $\mu\vdash n-1$ be obtained from $\lambda$ by removing a box at $\beta_k$. Then, we have~\cite[Lemma 3.3]{kerov2000anisotropic}
\begin{equation}\label{eq-6200309}
\frac{|\Tab(\mu)|}{|\Tab(\lambda)|}=\frac{(\alpha_L-\beta_k)(\beta_k-\alpha_1)}{n}\prod_{i=1}^{k-1}\frac{\beta_k-\alpha_{i+1}}{\beta_k-\beta_i}\prod_{i=k+1}^{L-1}\frac{\beta_k-\alpha_i}{\beta_k-\beta_i}.
\end{equation}

\subsection{Schur-Weyl duality}
Let $\mathcal{H}_1,\mathcal{H}_2,\ldots,\mathcal{H}_n$ be a sequence of Hilbert spaces such that $\mathcal{H}_i\cong\mathbb{C}^d$ for $1\leq i\leq n$.
Consider the Hilbert space \(\bigotimes_{i=1}^n\mathcal{H}_i\). This space admits representations of the symmetric group \(\mathfrak{S}_n\) (i.e., the group of all permutations on the set $\{1,2,\ldots,n\}$) and unitary group \(\mathbb{U}_d\) (i.e., the group of unitaries on $d$-dimensional Hilbert space). The unitary group acts by simultaneous ``rotation'' as \(U^{\otimes n}\) for any \(U\in \mathbb{U}_d\) and the symmetric group acts by permuting tensor factors:
\begin{equation}
P(\pi)\ket{\psi_1}\cdots\ket{\psi_n}=\ket{\psi_{\pi^{-1}(1)}}\cdots\ket{\psi_{\pi^{-1}(n)}},
\end{equation}
where \(\pi\in\mathfrak{S}_n\). 
\begin{remark}
In this paper, when we say the symmetric group $\mathfrak{S}_n$ acts on $\mathcal{H}_{i_1}\otimes\mathcal{H}_{i_2}\otimes\cdots\otimes\mathcal{H}_{i_n}$, the order of tensor products matters instead of the indices $i_j$ of the Hilbert spaces. For example, if $\mathfrak{S}_3$ acts on $\mathcal{H}_1\otimes\mathcal{H}_3\otimes\mathcal{H}_2$ and let $\ket{\psi_1}\ket{\psi_3}\ket{\psi_2}\in\mathcal{H}_1\otimes\mathcal{H}_3\otimes\mathcal{H}_2$, then we have
\[P(\pi_{(13)})\ket{\psi_1}\ket{\psi_3}\ket{\psi_2}=\ket{\psi_2}\ket{\psi_3}\ket{\psi_1},\]
where $\pi_{(13)}\in\mathfrak{S}_3$ is the permutation that swaps integers $1$ and $3$.
\end{remark}
Two actions \(U^{\otimes n}\) and \(P(\pi)\) commute with each other, and hence $\bigotimes_{i=1}^n\mathcal{H}_i$ admits a representation of group \(\mathbb{U}_d\times \mathfrak{S}_n\). More specifically, the Schur-Weyl duality (see, e.g., \cite{fulton2013representation}) states that
\begin{equation}\label{eq-1111224}
\bigotimes_{i=1}^n\mathcal{H}_i\stackrel{\mathbb{U}_d\times \mathfrak{S}_n}{\cong}\bigoplus_{\lambda \vdash_d\, n}\mathcal{Q}_\lambda\otimes \mathcal{P}_\lambda,
\end{equation}
where \(\mathcal{Q}_\lambda\) and \(\mathcal{P}_\lambda\) are irreducible representations of \(\mathbb{U}_d\) and \(\mathfrak{S}_n\) labeled by Young diagram $\lambda$, respectively. 
\begin{remark}
In general, the representation $\mathcal{Q}_\lambda$ depends on $d$ and is sometimes denoted by $\mathcal{Q}_\lambda^d$. However, since the dimension $d$ is fixed throughout this work, we can omit it.
\end{remark}
\begin{remark}
We take the convention that $\mathcal{Q}_\lambda=0$ (the trivial space) if $\ell(\lambda)>d$.
Therefore, the summation in \cref{eq-1111224} might be simplified to $\bigoplus_{\lambda\vdash n} \mathcal{Q}_\lambda\otimes\mathcal{P}_\lambda$.
\end{remark}

Now, suppose we have two sequences of Hilbert spaces $(\mathcal{H}^\textup{A}_1,\ldots,\mathcal{H}^{\textup{A}}_n)$ and $(\mathcal{H}^{\textup{B}}_1,\ldots,\mathcal{H}^{\textup{B}}_m)$, where $\mathcal{H}^\textup{A}_i\cong \mathcal{H}^{\textup{B}}_j\cong \mathbb{C}^d$. We define the action of group $\mathfrak{S}_{n}\times\mathfrak{S}_m$ on $\bigotimes_{i=1}^n \mathcal{H}_i^\textup{A}\otimes \bigotimes_{j=1}^m \mathcal{H}_j^\textup{B}$ as $P(\pi^\textup{A})\otimes P(\pi^\textup{B})$ for $(\pi^\textup{A},\pi^\textup{B})\in\mathfrak{S}_n\times\mathfrak{S}_m$.
Similarly, we define the action of \(\mathbb{U}_d\times \mathbb{U}_d\) on \(\bigotimes_{i=1}^n \mathcal{H}^\textup{A}_i\otimes\bigotimes_{j=1}^m \mathcal{H}^\textup{B}_j\) as $(U^{\textup{A}})^{\otimes n}\otimes (U^{\textup{B}})^{\otimes m}$ for $(U^{\textup{A}},U^{\textup{B}})\in \mathbb{U}_d\times\mathbb{U}_d$. Note that the action of $\mathfrak{S}_n\times\mathfrak{S}_m$ commutes with the action of $\mathbb{U}_d\times\mathbb{U}_d$. Therefore, we have a generalized Schur-Weyl duality on this bipartite system as
\[\bigotimes_{i=1}^n \mathcal{H}^\textup{A}_i\otimes\bigotimes_{j=1}^m \mathcal{H}^\textup{B}_j\,\, \stackrel{\mathbb{U}_d\times\mathbb{U}_d\times \mathfrak{S}_{n}\times\mathfrak{S}_m}{\cong}\,\,\bigoplus_{\substack{\lambda\vdash_{d} \, n\\ \mu\vdash_{d} \, m}} \mathcal{Q}_\lambda\otimes\mathcal{Q}_\mu\otimes \mathcal{P}_{\lambda}\otimes\mathcal{P}_\mu,\]
where $\mathcal{Q}_\lambda\otimes\mathcal{Q}_\mu\otimes\mathcal{P}_\lambda\otimes\mathcal{P}_\mu$ is an irreducible representation of $\mathbb{U}_d\times\mathbb{U}_d\times\mathfrak{S}_n\times\mathfrak{S}_m$. 

\subsection{Young-Yamanouchi basis}
\paragraph{Branching rule of $\mathfrak{S}_n$.}
We now look deeper into the structure of $\mathcal{P}_\lambda$. Consider the following chain of groups:
\[\mathfrak{S}_1\subset \mathfrak{S}_2\subset\cdots\subset\mathfrak{S}_n,\]
where the embedding $\mathfrak{S}_i\rightarrow \mathfrak{S}_{i+1}$ is by identifying each $\pi\in\mathfrak{S}_i$ as the permutation in $\mathfrak{S}_{i+1}$ that fixes the element $i+1$. Note that the irreducible representations of $\mathfrak{S}_i$ are in one-to-one correspondence to the Young diagrams consisting of $i$ boxes.
Let $\lambda\vdash i$ be a Young diagram of $i$ boxes and $\mathcal{P}_\lambda$ be an irreducible representation of $\mathfrak{S}_i$. Then, $\mathcal{P}_\lambda$ is also a representation of $\mathfrak{S}_{i-1}$, which decomposes as a direct sum of irreducible representations of $\mathfrak{S}_{i-1}$. More specifically, we have the following branching rule~\cite{ceccherini2010representation,meliot2017representation}:
\begin{equation}\label{eq-682336}
\textup{Res}^{\mathfrak{S}_{i}}_{\mathfrak{S}_{i-1}}\mathcal{P}_\lambda=\bigoplus_{\substack{\mu: \mu\nearrow\lambda}} \mathcal{P}_\mu,
\end{equation}
where $\textup{Res}^{\mathfrak{S}_i}_{\mathfrak{S}_{i-1}}$ denotes the restriction of representation from $\mathfrak{S}_i$ to $\mathfrak{S}_{i-1}$, $\mu\nearrow \lambda$ means that $\lambda$ can be obtained from $\mu$ by adding one box. \cref{eq-682336} means that we can decompose the vector space $\mathcal{P}_\lambda$ into direct sum of smaller vector spaces. By iterating this process (i.e., taking further restrictions on $\mathcal{P}_\mu$), we finally obtain a direct sum of the trivial one-dimensional representations of $\mathfrak{S}_1$. Therefore, this process determines a basis (up to scalar factors) of $\mathcal{P}_\lambda$ and each basis vector can be parameterized by its branching path:
\begin{equation}\label{eq-690205}
\lambda^{(1)}\rightarrow \lambda^{(2)}\rightarrow \cdots\rightarrow \lambda^{(n)},
\end{equation}
where $\lambda^{(n)}=\lambda$, $\lambda^{(i)}\nearrow \lambda^{(i+1)}$, and $\lambda^{(1)}$ is the one-box Young diagram. This basis is called \textit{Young-Yamanouchi basis} or \textit{Gelfand-Tsetlin basis} for $\mathfrak{S}_n$~\cite{okounkov1996new}. For convenience, we will simply call it Young basis in this paper.

Furthermore, we can see that the path in \cref{eq-690205} uniquely corresponds to a standard Young tableau $T$ of shape $\lambda=\lambda^{(n)}$, where the integer $i$ is assigned to the box $\lambda^{(i)}\setminus\lambda^{(i-1)}$, in which $\lambda^{(i)}\setminus\lambda^{(i-1)}$ denotes the shape by removing from the shape of $\lambda^{(i)}$ all the boxes belonging to $\lambda^{(i-1)}$. By this mean, we may also denote the standard Young tableau as 
\[T=\lambda^{(1)}\rightarrow\lambda^{(2)}\rightarrow\cdots\rightarrow\lambda^{(n)}.\]
By choosing an appropriate inner product making $\mathcal{P}_\lambda$ a unitary representation of $\mathfrak{S}_n$ and then normalizing the Young basis w.r.t. the inner product, we obtain an orthonormal basis $\{\ket{T}\}_{T\in\Tab(\lambda)}$ of $\mathcal{P}_\lambda$ labeled by standard Young tableaux $T\in\Tab(\lambda)$. This also implies that $\dim(\mathcal{P}_\lambda)=|\Tab(\lambda)|$.

\paragraph{Isotypic component projectors.} Consider the group algebra $\mathbb{C}\mathfrak{S}_n$, which is the associative algebra containing formal linear combinations of the elements in $\mathfrak{S}_n$ with coefficients in $\mathbb{C}$. Note that any representation of $\mathfrak{S}_n$ is naturally a representation of $\mathbb{C}\mathfrak{S}_n$ by linearly extending the group action of $\mathfrak{S}_n$. 
For any $1\leq k\leq n$ and $\lambda\vdash k$, we define the following element
\begin{equation}\label{eq-6122124}
e_{\lambda}\coloneqq \frac{\dim(\mathcal{P}_\lambda)}{k!}\sum_{\pi\in\mathfrak{S}_k} \chi^*_{\lambda}(\pi) \pi\in\mathbb{C}\mathfrak{S}_k\subseteq \mathbb{C}\mathfrak{S}_n,
\end{equation}
where $\chi_\lambda(\cdot)$ is the character of $\mathcal{P}_\lambda$.
By the orthogonality of characters (see e.g., \cite[Sec. 4.5]{etingof2011introduction}), $e_\lambda$ acts on $\mathcal{P}_\lambda$ as the identity and acts on $\mathcal{P}_\mu$ as the null map for $\mu\vdash k,\mu\neq \lambda$, thus is the projector onto the isotypic component of $\lambda$ (in fact, for any unitary representation of $\mathfrak{S}_k$, $e_\lambda$ acts as an orthogonal projector). Therefore, for any standard Young tableau $T=\lambda^{(1)}\rightarrow\cdots\rightarrow\lambda^{(n)}$, we have
\begin{equation}\label{eq-6110041}
e_{\lambda^{(k)}}\ket{T}=\ket{T},
\end{equation}
for any $1\leq k\leq n$. Moreover, $\ket{T}$ is the only vector (up to a scalar) satisfying \cref{eq-6110041} for any $1\leq k\leq n$. This means, under the algebra isomorphism  $\mathbb{C}\mathfrak{S}_n\cong \bigoplus_{\mu\vdash n}\mathcal{L}(\mathcal{P}_\mu)$ (due to the semisimplicity of $\mathbb{C}\mathfrak{S}_n$), we can write
\begin{equation}\label{eq-61301399}
\ketbra{T}{T}=\prod_{k=1}^n e_{\lambda^{(k)}},
\end{equation}
where $\ketbra{T}{T}$ is considered as a linear operator in $\bigoplus_{\mu\vdash n}\mathcal{L}(\mathcal{P}_\mu)$. 

\paragraph{Young's orthogonal form.}
For $1\leq i\leq n-1$, let $s_i\in\mathfrak{S}_n$ be the $i$-th adjacent transposition (i.e., $s_i$ swaps integers $i$ and $i+1$). It is known that those $s_i$ generate the group $\mathfrak{S}_n$ and are called Coxeter generators. The Young's orthogonal form~\cite{ceccherini2010representation} gives an explicit description of the action of $s_i$ on the Young basis:
\begin{equation}\label{eq-6140338}
s_i\ket{T}=\frac{1}{r(T)}\ket{T}+\sqrt{1-\frac{1}{r(T)^2}}\ket{s_iT},
\end{equation}
where $r(T)$ is the axial distance from the $(i+1)$-box to the $i$-box in $T$, $s_iT$ is the Young tableau obtained from $T$ by swapping the $i$-box and the $(i+1)$-box.\footnote{Note that $s_i T$ may not be a standard Young tableau. Nevertheless, in such case $\sqrt{1-1/r(T)^2}=0$.}

\section{Hardness of unitary time-reversal}
\label{sec:hardness_of_time_reversal}

In this section, we will prove our lower bounds for quantum unitary time-reversal. First, we provide a tight lower bound for time-reversal with the average-case distance error (see \cref{def-6290002}).
\begin{theorem}\label{thm-640254}
Suppose there is an algorithm that approximately implements $U^{-1}$ to within the average-case distance error $\epsilon$, using $n$ queries to the unknown $d$-dimensional unitary $U$. Then, it must satisfy $n\geq d(d+1)(1-\epsilon)-(d+1)$.
\end{theorem}
The proof is deferred to \cref{sec-6290007}.
Note that this directly implies a tight lower bound for time-reversal with diamond norm error.
\begin{corollary}\label{coro-6290408}
If we change the average-case distance in \cref{thm-640254} to diamond norm, then we have $n\geq d^2(1-\epsilon)-1$.
\end{corollary}
\begin{proof}
By \cref{eq-6282353}, we know that any $\epsilon$-approximate algorithm in diamond norm is an $(\frac{d}{d+1}\epsilon)$-approximate algorithm in the average-case distance. Therefore, the lower bound in \cref{thm-640254} implies a lower bound $d^2(1-\epsilon)-1$ for any diamond norm algorithm.
\end{proof}

From \Cref{coro-6290408}, we can further obtain a tight lower bound for the generalized unitary time-reversal task. 
Specifically, the task is to approximate the time-reverse unitary $U^{-t}\coloneqq e^{-iHt}$ given a reversing time $t>0$, where $H$ is allowed to be any specific Hamiltonian that satisfies $e^{iH}=e^{i\theta} U$ for some real number $\theta$, and $\norm*{H}\leq \pi$, in which $\|\cdot\|$ is the operator norm. 

\begin{corollary}
    \label{cor:generalized-time-reversal}
    Suppose there is an algorithm that approximately implements $U^{-t}$ to within constant diamond norm error $\epsilon\leq 10^{-5}$ for constant reversing time $t\geq 0.1$, using $n$ queries to the unknown $d$-dimensional unitary $U$. Then, it must satisfy $n= \Omega\parens*{d^2}$.
\end{corollary}

Here, the ranges $t\geq 0.1$ and $\epsilon\leq 10^{-5}$ are rather loose and chosen for simplicity of proof. 
The proof of \Cref{cor:generalized-time-reversal} is deferred to \Cref{sub:generalized_time_reversal}.

\subsection{Main framework}\label{sec-6290007}
In this subsection, we introduce our main framework and settings for the proof of \cref{thm-640254}.
Let $d\geq 2$ and $n\geq 1$ be arbitrary but fixed integers. We define a sequence of $2n+3$ Hilbert spaces $\mathcal{H}_0, \mathcal{H}_1,\ldots,\mathcal{H}_{2n+2}$ such that $\mathcal{H}_i\cong\mathbb{C}^d$ for $0\leq i\leq 2n+2$. 

\begin{figure}
    \centering
    \includegraphics[width=0.95\linewidth]{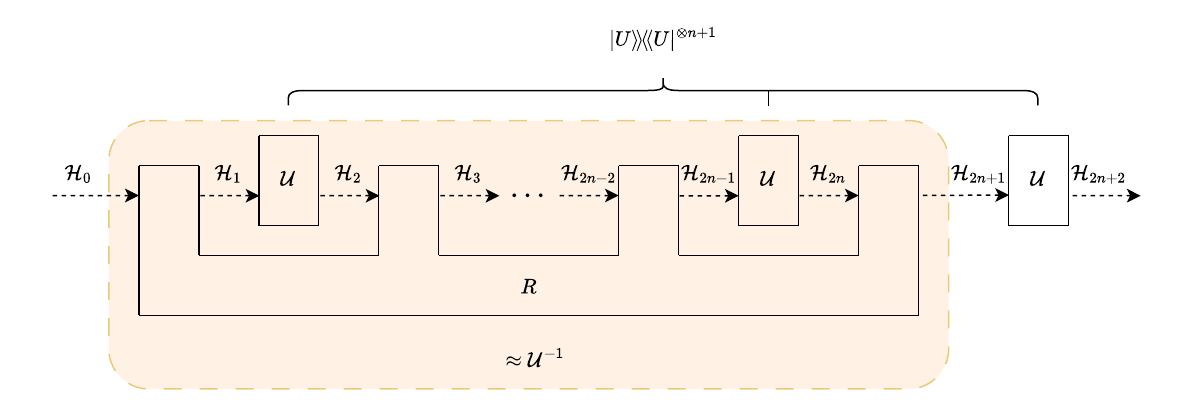}
    \caption{The quantum comb $R$ is a unitary time-reversal algorithm with error bounded by $\epsilon$ using $n$ queries to $U$ and the overall channel from $\mathcal{H}_0$ to $\mathcal{H}_{2n+2}$ approximates the identity channel $\mathcal{I}$, i.e., $R\star C_U \stackrel{\epsilon}{\approx} \kettbbra{I}{I}$.}
    \label{fig-6210108}
\end{figure}

Any algorithm, which takes $n$ queries to an unknown $d$-dimensional unitary $U$ and produces a $d$-dimensional quantum channel, can be described as an $(n+1)$-comb on $(\mathcal{H}_0,\mathcal{H}_1,\ldots,\mathcal{H}_{2n+1})$ where the $i$-th query is inserted on $(\mathcal{H}_{2i-1},\mathcal{H}_{2i})$, and the produced quantum channel is from $\mathcal{H}_{0}$ to $\mathcal{H}_{2n+1}$ (see the orange region in \cref{fig-6210108}).
Suppose $R$ is a unitary time-reversal algorithm with average-case distance error bounded by $\epsilon$. We make an additional query to $U$ at the end of the algorithm $R$ (see \cref{fig-6210108}).
Then, the overall channel from $\mathcal{H}_0$ to $\mathcal{H}_{2n+2}$ should approximate the identity channel.

To formalize this, we first define the following $(n+1)$-comb $C_U$, which represents the $n+1$ queries to $U$.
\begin{definition}\label{def-690411}
For any unitary $U\in\mathbb{U}_d$, we define the $(n+1)$-comb $C_U$ on $(\mathcal{H}_1,\mathcal{H}_2,\ldots,\mathcal{H}_{2n+2})$ as
\[C_U\coloneqq \underbrace{\kettbbra{U}{U}\otimes \cdots\otimes \kettbbra{U}{U}}_{n+1},\]
where the $i$-th tensor factor $\kettbbra{U}{U}$ is the $1$-comb on $(\mathcal{H}_{2i-1},\mathcal{H}_{2i})$ corresponding to the unitary channel $\mathcal{U}$.
\end{definition} 

Therefore, we know that $R\star C_U\stackrel{\epsilon}{\approx} \kettbbra{I}{I}$. Furthermore, when the unknown unitary $U$ is randomly sampled from the Haar measure, the overall channel should still approximate the identity channel  i.e., $R\star \E_U[C_U] \stackrel{\epsilon}{\approx} \kettbbra{I}{I}$. Specifically, we have:
\begin{proposition}\label{fact-641442}
If an $(n+1)$-comb $R$ on $(\mathcal{H}_0,\mathcal{H}_1,\ldots,\mathcal{H}_{2n+1})$ is a time-reversal algorithm with the average-case distance error bounded by $\epsilon$, then we have
\[\tr\!\left( \kettbbra{I}{I}\cdot \left(R\star \E_U[C_{U}]\right)\right)\geq d^2-d(d+1)\epsilon,\]
where $\kettbbra{I}{I}$ is the $1$-comb on $(\mathcal{H}_0,\mathcal{H}_{2n+2})$ corresponding to the identity channel, 
$\star$ is the link product (see \cref{def-720255}),
$\E_U$ denotes the expectation over a Haar random unitary $U\in\mathbb{U}_d$.
\end{proposition}
\begin{proof}
For any $U\in\mathbb{U}_d$, $R$ can produce a quantum channel that approximates $\mathcal{U}^{-1}$ to within the average-case distance error $\epsilon$. Then, we apply the unitary channel $\mathcal{U}$ after the channel produced by $R$ (see \cref{fig-6210108}). Due to the unitary invariance of the average-case distance (see \cref{eq-6292334}), we obtain a channel that approximates the identity channel $\mathcal{I}$ to within the average-case distance error $\epsilon$, i.e.,
\[\Davg(\mathcal{E}_{R\star C_U}, \mathcal{I})\leq \epsilon,\]
where $\mathcal{E}_{R\star C_U}$ denotes the quantum channel corresponding to the $1$-comb (Choi state) $R\star C_U$.
By \cref{eq-6282338}, this means
\[\frac{1}{d^2}\tr\!\left( \kettbbra{I}{I}\cdot \left(R\star C_{U}\right)\right)\geq 1-\frac{d+1}{d}\epsilon.\]
Then, by taking the Haar expectation over $U\in\mathbb{U}_d$ and noting that $\star$ is bilinear, we have
\[\tr\!\left( \kettbbra{I}{I}\cdot \left(R\star  \E_U[C_{U}]\right)\right)\geq d^2-d(d+1)\epsilon.\]
\end{proof}

On the other hand, we prove the following result, showing that when $U$ is randomly sampled from the Haar measure and the number of queries $n$ is not too large, the overall channel must be very depolarizing. The proof is deferred to \cref{sec-641554}.
\begin{theorem}\label{thm-691837}
For any $(n+1)$-comb $X$ on $(\mathcal{H}_0,\mathcal{H}_1,\ldots,\mathcal{H}_{2n+1})$, we have 
\[X\star \E_{U}[C_U] \sqsubseteq \frac{n+1}{d}\cdot I_{\mathcal{H}_{2n+2}}\otimes I_{\mathcal{H}_0},\]
where $\E_U$ denotes the expectation over a Haar random unitary $U\in\mathbb{U}_d$.
\end{theorem}

Now, we are able to prove our main query lower bound (i.e., \cref{thm-640254}), assuming \cref{thm-691837}.
\begin{proof}[Proof of \cref{thm-640254}]
Suppose $R$ is an algorithm that approximately implements $U^{-1}$ to within the average-case distance error $\epsilon$, using $n$ queries to $U$.
Then, $R$ can be described as an $(n+1)$-comb on $(\mathcal{H}_0,\mathcal{H}_1,\ldots,\mathcal{H}_{2n+1})$. 
Combining \cref{fact-641442} and \cref{thm-691837}, we know that
\[n+1\geq \tr\!\left( \kettbbra{I}{I}\cdot \left(R\star \E_U[C_{U}]\right)\right) \geq d^2-d(d+1)\epsilon,\]
which means $n\geq d(d+1)(1-\epsilon)-(d+1)$.
\end{proof}

Equivalent to the query lower bound, we can also give an upper bound for the average fidelity achievable by an $n$-query unitary time-reversal algorithm for any fixed $n$, still using \cref{thm-691837}.
Specifically, let $R(U)$ denote the channel produced by the algorithm $R$ applied on $n$ queries to $U$. 
The average performance of $R$ can be evaluated using the channel fidelity between $R(U)$ and $\mathcal{U}^{-1}$ for Haar random $U$, i.e.,
\[F(R)\coloneqq \E_U\!\left[F_\textup{ch}(R(U),\mathcal{U}^{-1})\right].\]
Then, we can see the following result.
\begin{theorem}\label{thm-290047}
For any $n$-query algorithm $R$, we have $F(R)\leq \frac{n+1}{d^2}$.
\end{theorem}
\begin{proof}
Note that
\begin{align}
\E_U\!\left[F_{\textup{ch}}(R(U),\mathcal{U}^{-1})\right]&=\E_U\!\left[F_{\textup{ch}}(\mathcal{U}\circ R(U),\mathcal{I})\right] \label{eq-282308} \\
&=\frac{1}{d^2}\tr(R\star \E_U[C_U] \cdot \kettbbra{I}{I})\nonumber\\
&\leq \frac{n+1}{d^3} \tr(\kettbbra{I}{I}) \label{eq-282249}\\
&=\frac{n+1}{d^2},\nonumber
\end{align}
where \cref{eq-282308} is due to unitary invariance of the channel fidelity, and \cref{eq-282249} is due to \cref{thm-691837} and the fact that $R$ is an $(n+1)$-comb.
\end{proof}
We note that this upper bound matches the optimal fidelities numerically obtained by the semidefinite programming in \cite[Table I]{grinko2025sequential}.
This provides an additional empirical illustration of the asymptotic tightness established above, and further suggests that the fidelity upper bound in \cref{thm-290047} (or equivalently, the query lower bound in \cref{thm-640254}) is likely optimal even in a non-asymptotic sense (i.e., without suppressing constant factors).

\subsection{Proof of \texorpdfstring{\cref{thm-691837}}{Theorem 3.6}}\label{sec-641554}
Here, we define the following linear operators $A_k$ for $1\leq k\leq n$, which we refer to as \textit{stair operators} and use as a tool in analyzing the Haar moment $\E_U[C_U]$.
\begin{definition}[Stair operators]\label{def-6100203}
For each $1\leq k \leq n$, we define the linear operator $A_{k}$ on 
\[\left(\left(\bigotimes_{i=1}^{k}\mathcal{H}_{2i}\right) \otimes \mathcal{H}_{2n+2}\right) \otimes \left(\bigotimes_{i=1}^k\mathcal{H}_{2i-1}\right) \stackrel{\mathbb{U}_d\times\mathbb{U}_d\times\mathfrak{S}_{k+1}\times\mathfrak{S}_k}{\cong} \bigoplus_{\substack{\lambda\vdash_d k+1\\\mu\vdash_d k}}\mathcal{Q}_\lambda\otimes \mathcal{Q}_\mu\otimes\mathcal{P}_\lambda\otimes\mathcal{P}_\mu,\]
where $\mathfrak{S}_{k+1}$ acts on $\mathcal{H}_2\otimes\mathcal{H}_4\otimes\cdots\otimes\mathcal{H}_{2k}\otimes\mathcal{H}_{2n+2}$ and $\mathfrak{S}_k$ acts on $\mathcal{H}_1\otimes\mathcal{H}_3\otimes\cdots\otimes\mathcal{H}_{2k-1}$,
as
\[A_k\coloneqq\bigoplus_{\lambda\vdash_d k+1} \bigoplus_{\mu:\mu\nearrow \lambda} \frac{\dim(\mathcal{P}_\lambda)}{\dim(\mathcal{P}_\mu)\dim(\mathcal{Q}_\lambda)}  I_{\mathcal{Q}_\lambda}\otimes I_{\mathcal{Q}_\mu} \otimes \sum_{\substack{T,S\in\Tab(\lambda,\mu)}}\ketbra{T}{S}\otimes \ketbra{T^{\downarrow}}{S^{\downarrow}},\]
where $T^{\downarrow}$ denotes the standard Young tableau obtained from $T$ by removing the box containing the largest integer, $\Tab(\lambda,\mu)$ denotes the set of standard Young tableaux $T\in\Tab(\lambda)$ such that $\Sh(T^\downarrow)=\mu$.
\end{definition}
Note that each stair operator $A_k$ is defined in an asymmetric bipartite system where one of the parties contains $k+1$ copies of $\mathbb{C}^d$ and the another contains $k$ copies of $\mathbb{C}^d$. 
In addition, we will demonstrate a contraction process of $A_k$ which leads to a smaller $A_{k-1}$.
As shown in \Cref{fig-751905},
each $A_k$ resembles a stair, and the contraction process resembles walking down a stair.
This is why we refer to $\{A_k\}_{k=1}^n$ as stair operators. 
A similar construction is adopted in \cite{YKS+24}.
However, their analysis applies only to the case of $n\leq d-1$, where a lower bound of $\Omega(d)$ is proved for the unitary time-reversal.
In contrast, our construction and analysis apply to the general case without any constraint on $n$.
Specifically, we provide the following lemmas to analyze the behavior of the Haar moment $\E_U[C_U]$ when combined with another quantum comb.
The proofs are deferred to \cref{sec-641552} and \cref{sec-641553}.

First, the stair operators provide a good upper bound (w.r.t. L\"owner order) for $\E_U[C_U]$, see also \cref{fig-751904}.
\begin{lemma}\label{lemma-641548}
Let $C_U$ be that defined in \cref{def-690411}. Then
\[\E_U[C_U]\sqsubseteq I_{\mathcal{H}_{2n+1}}\otimes A_n,\]
where $I_{\mathcal{H}_{2n+1}}$ is the identity operator on $\mathcal{H}_{2n+1}$.
\end{lemma}
Moreover, the stair operators share the similar tensor contraction property with quantum combs (see \cref{eq-681546}). Specifically, tracing out the ``second-last'' subsystem of a stair operator yields an operator bounded by the tensor product of the identity (with a coefficient) and a reduced stair operator, see also \cref{fig-751905}.
\begin{lemma}\label{lemma-641551}
For any $2\leq k\leq n$, we have
\begin{equation}\label{eq-641413}
\tr_{\mathcal{H}_{2k}}(A_k)\sqsubseteq \frac{k+1}{k}\cdot  I_{\mathcal{H}_{2k-1}}\otimes A_{k-1},
\end{equation}
where $I_{\mathcal{H}_{2k-1}}$ is the identity operator on $\mathcal{H}_{2k-1}$. Moreover, we have
\begin{equation}\label{eq-641403}
\tr_{\mathcal{H}_{2}}(A_1)=\frac{2}{d}\cdot I_{\mathcal{H}_{1}}\otimes I_{\mathcal{H}_{2n+2}}.
\end{equation}
\end{lemma}
\begin{figure}[t]
    \centering
    \hspace{-10mm}
    \begin{subfigure}[b]{0.54\linewidth}
    \centering
    \includegraphics[width=1.0\linewidth]{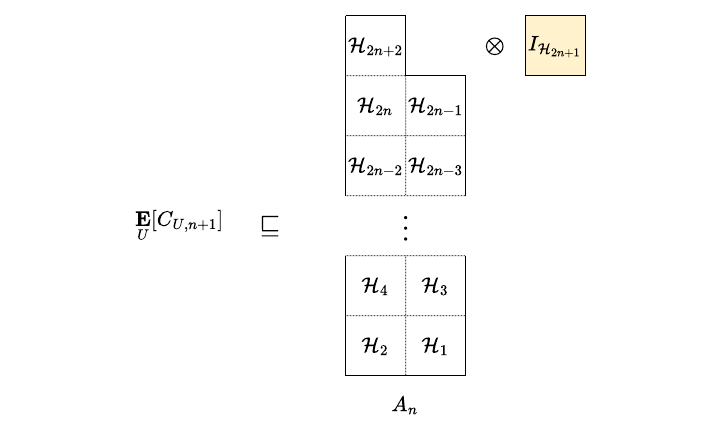}
    \vspace{-6mm}
    \caption{Bounding the Haar moment $\E_U[C_U]$.}\label{fig-751904}
    \end{subfigure}
    \hspace{-10mm}
    \begin{subfigure}[b]{0.54\linewidth}
    \centering
    \includegraphics[width=1.0\linewidth]{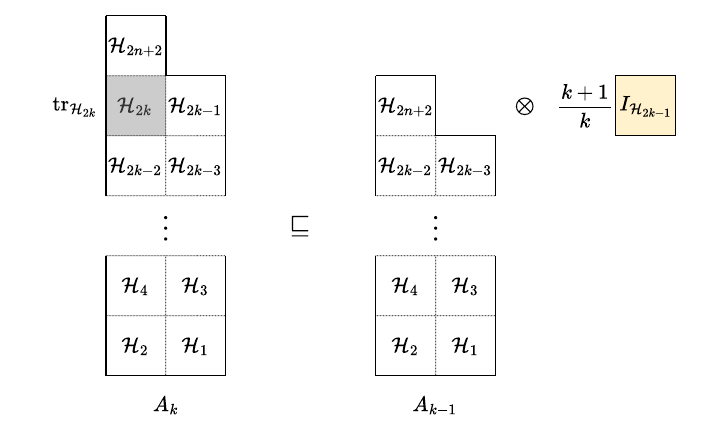}
    \vspace{-6mm}
    \caption{Contraction of the stair operators.}\label{fig-751905}
    \end{subfigure}
    \vspace{0mm}
    \caption{Properties of the stair operators $\{A_k\}_{k=1}^n$ (see \cref{lemma-641548} and \cref{lemma-641551}).}\label{fig-751907}
\end{figure}

\cref{lemma-641551} implies that when a stair operator is combined with another quantum comb, they can be sequentially contracted. The following corollary is a direct application of \cref{lemma-641551}, which characterizes the behavior of the stair operator $A_k$ when combined with an arbitrary $(k+1)$-comb. 
\begin{corollary}\label{lemma-641443}
For any $1\leq k\leq n$, and any $(k+1)$-comb $X$ on $(\mathcal{H}_0,\mathcal{H}_1,\ldots,\mathcal{H}_{2k+1})$, we have
\[X\star (I_{\mathcal{H}_{2k+1}}\otimes A_k)\sqsubseteq \frac{k+1}{d} \cdot I_{\mathcal{H}_{2n+2}}\otimes I_{\mathcal{H}_0}.\]
\end{corollary}
\begin{proof}
We use induction on $k$. For the case $k=1$, we have
\begin{align}
X\star (I_{\mathcal{H}_3}\otimes A_1) 
&=\tr_{\mathcal{H}_{1:3}}\!\left(X^{\textup{T}_{\mathcal{H}_{1:3}}}\cdot (I_{\mathcal{H}_{3}}\otimes A_1)\right)\nonumber\\
&=\tr_{\mathcal{H}_{1:2}}\!\left( \tr_{\mathcal{H}_3}(X)^{\textup{T}_{\mathcal{H}_{1:2}}}\cdot  A_1\right)\nonumber\\
&=\tr_{\mathcal{H}_{1:2}}\!\left(\left(I_{\mathcal{H}_{2}}\otimes X'\right)^{\textup{T}_{\mathcal{H}_{1:2}}}\cdot A_1\right)\label{eq-641400}\\
&=\tr_{\mathcal{H}_{1}}\left(\left(X'\right)^{\textup{T}_{\mathcal{H}_{1}}}\cdot \tr_{\mathcal{H}_{2}}(A_1)\right)\nonumber\\
&=\frac{2}{d} \cdot \tr_{\mathcal{H}_{1}}\left(\left(X'\right)^{\textup{T}_{\mathcal{H}_{1}}}\cdot I_{\mathcal{H}_{1}}\otimes I_{\mathcal{H}_{2n+2}}\right)\label{eq-641402}\\
&=\frac{2}{d}\cdot I_{\mathcal{H}_0}\otimes I_{\mathcal{H}_{2n+2}},\label{eq-641404}
\end{align}
where \cref{eq-641400} is because $X$ is a $2$-comb on $(\mathcal{H}_0,\mathcal{H}_1,\mathcal{H}_2,\mathcal{H}_3)$, \cref{eq-641402} is due to \cref{eq-641403} in \cref{lemma-641551}, and \cref{eq-641404} is because $X'$ is a $1$-comb on $(\mathcal{H}_0,\mathcal{H}_1)$.

Now, suppose the induction hypothesis holds for $k-1$, we prove it also hods for $k$.
\begin{align}
X\star (I_{\mathcal{H}_{2k+1}}\otimes A_k)
&=\tr_{\mathcal{H}_{1:2k+1}}\!\left(X^{\textup{T}_{\mathcal{H}_{1:2k+1}}}\cdot (I_{\mathcal{H}_{2k+1}}\otimes A_k)\right)\nonumber\\
&=\tr_{\mathcal{H}_{1:2k}}\!\left(\tr_{\mathcal{H}_{2k+1}}\!\left(X\right)^{\textup{T}_{\mathcal{H}_{1: 2k}}}\cdot A_k\right) \nonumber\\
&=\tr_{\mathcal{H}_{1:2k}}\!\left(\left(I_{\mathcal{H}_{2k}}\otimes X'\right)^{\textup{T}_{\mathcal{H}_{1:2k}}}\cdot A_k\right)\label{eq-641411}\\
&=\tr_{\mathcal{H}_{1:2k-1}}\left(\left(X'\right)^{\textup{T}_{\mathcal{H}_{1:2k-1}}}\cdot \tr_{\mathcal{H}_{2k}}(A_k)\right)\nonumber\\
&=\left(X'\right)\star \tr_{\mathcal{H}_{2k}}(A_k)\nonumber \\
&\sqsubseteq \frac{k+1}{k} \cdot \left(X'\right)\star (I_{\mathcal{H}_{2k-1}}\otimes A_{k-1})\label{eq-641412}\\
&\sqsubseteq \frac{k+1}{d} \cdot I_{\mathcal{H}_{2n+2}}\otimes I_{\mathcal{H}_0},\label{eq-641414}
\end{align}
where \cref{eq-641411} is because $X$ is an $(k+1)$-comb and thus $X'$ is a $k$-comb, \cref{eq-641412} is due to \cref{eq-641413} in \cref{lemma-641551}, and \cref{eq-641414} is due to the induction hypothesis for the case of $k-1$
\end{proof}

Then, we are able to prove \cref{thm-691837}, assuming \cref{lemma-641548} and \cref{lemma-641551}.
\begin{proof}[Proof of \cref{thm-691837}]
Note that
\begin{align}
X\star \E_U[C_U]
&\sqsubseteq X\star (I_{\mathcal{H}_{2n+1}}\otimes A_n) \label{eq-641547}\\
&\sqsubseteq\frac{n+1}{d} \cdot I_{\mathcal{H}_{2n+2}}\otimes I_{\mathcal{H}_0},\label{eq-641548}
\end{align}
where \cref{eq-641547} is by \cref{lemma-641548} combined with \cref{eq-721843}, and \cref{eq-641548} is by taking $k=n$ in \cref{lemma-641443}.
\end{proof}

\subsection{Bounding the Haar moment: proof of \texorpdfstring{\cref{lemma-641548}}{Lemma 3.9}}\label{sec-641552}
In this subsection, we bound the Haar moment $\E_U[C_U]$ using the stair operators (see \cref{lemma-641548}).

First, we need the following fact: 
\begin{align}
\E_{U}[C_{U}]
&=\bigoplus_{\lambda \vdash_d n+1} \frac{1}{\dim(\mathcal{Q}_\lambda)} I_{\mathcal{Q}_\lambda}\otimes I_{\mathcal{Q}_\lambda}\otimes \kettbbra{I_{\mathcal{P}_\lambda}}{I_{\mathcal{P}_\lambda}},\nonumber
\end{align}
which provides a representation-theoretic expression of the Haar moment $\E_U[C_U]$ (see \cref{prop-5121539}).
Then, we provide our proof of \cref{lemma-641548}.
\begin{proof}[Proof of \cref{lemma-641548}]
By \cref{prop-5121539}, we can express $\E_U[C_U]$ in the Schur-Weyl basis 
\[\bigoplus_{\substack{\lambda\vdash_d n+1\\ \mu\vdash_d n+1}}\mathcal{Q}_\lambda\otimes\mathcal{Q}_{\mu}\otimes\mathcal{P}_\lambda\otimes \mathcal{P}_\mu\stackrel{\mathbb{U}_d\times\mathbb{U}_d\times\mathfrak{S}_{n+1}\times\mathfrak{S}_{n+1}}{\cong} \bigotimes_{i=1}^{n+1}\mathcal{H}_{2i}\otimes\bigotimes_{i=1}^{n+1}\mathcal{H}_{2i-1}.\]
On the other hand, by \cref{def-6100203}, $A_n$ is defined on the Schur-Weyl basis
\[\bigoplus_{\substack{\lambda\vdash_d n+1\\ \mu\vdash_d n}}\mathcal{Q}_\lambda\otimes\mathcal{Q}_{\mu}\otimes\mathcal{P}_\lambda\otimes \mathcal{P}_\mu\stackrel{\mathbb{U}_d\times\mathbb{U}_d\times\mathfrak{S}_{n+1}\times\mathfrak{S}_{n}}{\cong} \bigotimes_{i=1}^{n+1}\mathcal{H}_{2i}\otimes\bigotimes_{i=1}^{n}\mathcal{H}_{2i-1}.\]
By comparing these two forms, we note that it is natural to study that for $\mu\vdash n$ and $T,S\in\Tab(\mu)$, how the linear operator 
\(I_{\mathcal{H}_{2n+1}}\otimes I_{\mathcal{Q}_\mu}\otimes \ketbra{T}{S}\in\mathcal{L}(\mathcal{H}_{2n+1}\otimes\mathcal{Q}_\mu\otimes\mathcal{P}_\mu)\subseteq\mathcal{L}(\bigotimes_{i=1}^{n+1}\mathcal{H}_{2i-1})\) is expressed as an operator on $\bigoplus_{\mu\vdash_d n+1}\mathcal{Q}_{\mu}\otimes\mathcal{P}_\mu$. 
For this, we have the following result by using \cref{fact-6101953}:
\[I_{\mathcal{H}_{2n+1}}\otimes I_{\mathcal{Q}_\mu}\otimes \ketbra{T}{S}=\bigoplus_{\substack{\nu:\mu\nearrow \nu}} I_{\mathcal{Q}_\nu}\otimes \ketbra{T^{\uparrow \nu}}{S^{\uparrow \nu}},\]
where $T^{\uparrow \nu}$ denotes the standard Young tableau obtained from $T$ by adding a box filled with $n+1$, resulting in the shape $\nu$.
This means 
\begin{align}
I_{\mathcal{H}_{2n+1}}\otimes A_n
&=\bigoplus_{\lambda\vdash_d k+1} \bigoplus_{\mu:\mu\nearrow \lambda} \frac{\dim(\mathcal{P}_\lambda)}{\dim(\mathcal{P}_\mu)\dim(\mathcal{Q}_\lambda)}  I_{\mathcal{Q}_\lambda}\otimes \bigoplus_{\nu:\mu\nearrow\nu} I_{\mathcal{Q}_\nu} \otimes \sum_{\substack{T,S\in\Tab(\lambda,\mu)}}\ketbra{T}{S}\otimes\ketbra{(T^{\downarrow})^{\uparrow\nu}}{(S^\downarrow)^{\uparrow\nu}} \nonumber\\
&\sqsupseteq \bigoplus_{\lambda\vdash_d k+1} \bigoplus_{\mu:\mu\nearrow \lambda} \frac{\dim(\mathcal{P}_\lambda)}{\dim(\mathcal{P}_\mu)\dim(\mathcal{Q}_\lambda)}  I_{\mathcal{Q}_\lambda}\otimes I_{\mathcal{Q}_\lambda} \otimes \sum_{\substack{T,S\in\Tab(\lambda,\mu)}}\ketbra{T}{S}\otimes\ketbra{T}{S} \label{eq-6110344}\\
&=\bigoplus_{\lambda\vdash_d k+1}  \frac{1}{\dim(\mathcal{Q}_\lambda)}  I_{\mathcal{Q}_\lambda}\otimes I_{\mathcal{Q}_\lambda} \otimes \sum_{\mu:\mu\nearrow \lambda} \frac{\dim(\mathcal{P}_\lambda)}{\dim(\mathcal{P}_\mu)} \ketbra{\Phi_\mu^\lambda}{\Phi_\mu^\lambda} \label{eq-6110345}\\
&\sqsupseteq \bigoplus_{\lambda\vdash_d k+1}  \frac{1}{\dim(\mathcal{Q}_\lambda)}  I_{\mathcal{Q}_\lambda}\otimes I_{\mathcal{Q}_\lambda} \otimes \left(\sum_{\mu:\mu\nearrow \lambda} \ket{\Phi_\mu^\lambda}\right)\left(\sum_{\mu:\mu\nearrow \lambda} \bra{\Phi_\mu^\lambda}\right)\label{eq-6110351}\\
&= \bigoplus_{\lambda\vdash_d k+1}  \frac{1}{\dim(\mathcal{Q}_\lambda)}  I_{\mathcal{Q}_\lambda}\otimes I_{\mathcal{Q}_\lambda} \otimes \kettbbra{I_{\mathcal{P}_\lambda}}{I_{\mathcal{P}_\lambda}}\label{eq-6110410}\\
&=\E_U[C_U],\nonumber
\end{align}
where \cref{eq-6110344} is by discarding those $\nu$ such that $\nu\neq\lambda$; in \cref{eq-6110345}, $\ket{\Phi_{\mu}^\lambda}\coloneqq\sum_{\substack{T\in\Tab(\lambda,\mu)}}\ket{T}\ket{T}$;
\cref{eq-6110410} is because $\kett{I_{\mathcal{P}_\lambda}}=\sum_{T\in\Tab(\lambda)}\ket{T}\ket{T}=\sum_{\mu:\mu\nearrow\lambda}\ket{\Phi_{\mu}^\lambda}$;
\cref{eq-6110351} is by using \cref{fact-5122103} and the fact:
\begin{align}
&\left(\sum_{\mu:\mu\nearrow \lambda} \bra{\Phi_\mu^\lambda}\right)\left(\sum_{\mu:\mu\nearrow\lambda}\frac{\dim(\mathcal{P}_\lambda)}{\dim(\mathcal{P}_\mu)}\ketbra{\Phi_\mu^\lambda}{\Phi_\mu^\lambda}\right)^{-1}\left(\sum_{\mu:\mu\nearrow \lambda} \ket{\Phi_\mu^\lambda}\right)\nonumber\\
=&\left(\sum_{\mu:\mu\nearrow \lambda} \bra{\Phi_\mu^\lambda}\right) \left(\sum_{\mu:\mu\nearrow\lambda}\frac{\dim(\mathcal{P}_\mu)}{\dim(\mathcal{P}_\lambda)}\cdot \frac{1}{\dim(\mathcal{P}_\mu)^2}\cdot\ketbra{\Phi_\mu^\lambda}{\Phi_\mu^\lambda}\right)\left(\sum_{\mu:\mu\nearrow \lambda} \ket{\Phi_\mu^\lambda}\right) \label{eq-760315}\\
=&\sum_{\mu:\mu\nearrow\lambda} \frac{\dim(\mathcal{P}_\mu)}{\dim(\mathcal{P}_\lambda)}\cdot\frac{1}{\dim(\mathcal{P}_\mu)^2}\cdot\dim(\mathcal{P}_\mu)^2 \nonumber \\
=&1, \label{eq-760316}
\end{align}
in which \cref{eq-760315} is because $\ket{\Phi_\mu^\lambda}$ are pairwise orthogonal and $|\Tab(\lambda,\mu)|=|\Tab(\mu)|=\dim(\mathcal{P}_\mu)$, \cref{eq-760316} is due to $\sum_{\mu:\mu\nearrow\lambda}\dim(\mathcal{P}_\mu)=\dim(\mathcal{P}_\lambda)$. 
\end{proof}

\subsection{Contraction of stair operators: proof of \texorpdfstring{\cref{lemma-641551}}{Lemma 3.10}}\label{sec-641553}
In this subsection, we prove the contraction properties of the stair operators (see \cref{lemma-641551}).

\subsubsection{The case of \texorpdfstring{$2\leq k\leq n$}{2<=k<=n}}
First, we prove for $2\leq k\leq n$,
\[\tr_{\mathcal{H}_{2k}}(A_k)\sqsubseteq\frac{k+1}{k}\cdot I_{\mathcal{H}_{2k-1}}\otimes A_{k-1}.\]
Note that in the definition of $A_k$ (see \cref{def-6100203}), the symmetric group $\mathfrak{S}_{k+1}$ acts on $\mathcal{H}_2\otimes\mathcal{H}_4\otimes\cdots\otimes\mathcal{H}_{2k}\otimes\mathcal{H}_{2n+2}$. 
We can not directly use \cref{lemma-6120245} since $\tr_{\mathcal{H}_{2k}}$ is tracing out the second-to-last subsystem $\mathcal{H}_{2k}$ instead of the last subsystem $\mathcal{H}_{2n+2}$. For this, we swap the last two tensor factors of $A_k$ by using the $k$-th adjacent transposition $s_{k}$. Specifically, let $P(\pi)$ be the action of $\pi\in\mathfrak{S}_{k+1}$ on $\mathcal{H}_2\otimes\mathcal{H}_4\otimes\cdots\otimes\mathcal{H}_{2k}\otimes\mathcal{H}_{2n+2}$. 
Then, $P(s_k)A_k P(s_k)$ is the operator on $\mathcal{H}_2\otimes\mathcal{H}_4\otimes\cdots\otimes\mathcal{H}_{2k-2}\otimes\mathcal{H}_{2n+2}\otimes\mathcal{H}_{2k}$ obtained from $A_k$ by swapping the positions of $\mathcal{H}_{2k}$ and $\mathcal{H}_{2n+2}$. 
By the Young's orthogonal form (see \cref{eq-6140338}), we have 
\begin{align}
P(s_k) A_k P(s_k)
&=\bigoplus_{\lambda\vdash_d k+1}\bigoplus_{\mu:\mu\nearrow \lambda}\frac{\dim(\mathcal{P}_\lambda)}{\dim(\mathcal{P}_\mu)\dim(\mathcal{Q}_\lambda)}I_{\mathcal{Q}_\lambda}\otimes I_{\mathcal{Q}_\mu}\otimes\sum_{\substack{T,S\in\Tab(\lambda,\mu)}} \bot_{T,S} \otimes\ketbra{T^\downarrow}{S^\downarrow},\label{eq-6142009}
\end{align}
where $\bot_{T,S}$ is defined as
\begin{equation}\label{eq-6142031}
\frac{1}{r(T)r(S)}\left(\ketbra{T}{S}+\sqrt{(r(T)^2-1)(r(S)^2-1)}\ketbra{s_kT}{s_kS}+\sqrt{r(T)^2-1}\ketbra{s_k T}{S}+\sqrt{r(S)^2-1}\ketbra{T}{s_k S}\right),
\end{equation}
where $r(T)$ is the axial distance from the $(k+1)$-box to the $k$-box in $T$.
Then, we can use \cref{lemma-6120245} to calculate the partial trace $\tr_{\mathcal{H}_{2k}}(P(s_k)A_k P(s_k))$.

To this end, let us consider each single summand in the RHS of \cref{eq-6142009}. Suppose $\lambda\vdash_d k+1$, $\mu\nearrow\lambda$, $T,S\in\Tab(\lambda,\mu)$. First note that it is not possible that $\Sh((s_kT)^\downarrow) = \Sh(S^\downarrow)$ since we know $\Sh(T^\downarrow)=\Sh(S^\downarrow)$. Thus by \cref{lemma-6120245}, we can ignore the terms $\ketbra{s_kT}{S}$ and $\ketbra{T}{s_kS}$ in \cref{eq-6142031} since they do not contribute after the partial trace. Next, we consider the terms $\ketbra{s_k T}{s_k S}$ and $\ketbra{T}{S}$ separately. 

\paragraph{Summands involving the term $\ketbra{s_kT}{s_k S}$.}
Consider the term $\ketbra{s_k T}{s_k S}$. We have
\begin{align}
&\tr_{\mathcal{H}_{2k}}\!\left(\frac{\dim(\mathcal{P}_\lambda)}{\dim(\mathcal{P}_\mu)\dim(\mathcal{Q}_\lambda)}I_{\mathcal{Q}_\lambda}\otimes I_{\mathcal{Q}_\mu} \otimes \frac{\sqrt{(r(T)^2-1)(r(S)^2-1)}}{r(T)r(S)}\ketbra{s_k T}{s_k S}\otimes\ketbra{T^\downarrow}{S^\downarrow}\right) \nonumber\\
\begin{split}
=& \frac{\dim(\mathcal{P}_\lambda)}{\dim(\mathcal{P}_\mu)\dim(\mathcal{Q}_{\Sh((s_kT)^\downarrow)})}I_{\mathcal{Q}_{\Sh((s_kT)^\downarrow)}}\otimes I_{\mathcal{Q}_\mu} \\
&\quad\quad\quad\quad\quad\quad\otimes \mathbbm{1}_{\Sh((s_kT)^\downarrow)=\Sh((s_kS)^\downarrow)}\cdot \frac{\sqrt{(r(T)^2-1)(r(S)^2-1)}}{r(T)r(S)}\ketbra{(s_k T)^\downarrow}{(s_k S)^{\downarrow}}\otimes\ketbra{T^\downarrow}{S^\downarrow}.\label{eq-6142152}
\end{split}
\end{align}
Then, we consider their summation.
For those summands of the form in \cref{eq-6142152}, by re-naming the variables $\Sh((s_k T)^{\downarrow})$ to $\nu$, $(s_k T)^\downarrow$ to $T$ and $(s_k S)^\downarrow$ to $S$, 
we can write the corresponding summation as
\begin{equation}\label{eq-6161627}
\bigoplus_{\substack{\nu\vdash_d k}}\frac{1}{\dim(\mathcal{Q}_\nu)}I_{\mathcal{Q}_\nu}\otimes \bigoplus_{\substack{\lambda:\nu\nearrow\lambda\\ \ell(\lambda)\leq d}} \bigoplus_{\substack{\mu:\mu\nearrow\lambda\\ \mu\neq\nu}}\frac{\dim(\mathcal{P}_\lambda)}{\dim(\mathcal{P}_\mu)}I_{\mathcal{Q}_\mu} \otimes \sum_{\substack{T,S\in\Tab(\nu,\mu\cap\nu)}}\frac{\sqrt{(r(T^{\uparrow \lambda})^2-1)(r(S^{\uparrow \lambda})^2-1)}}{r(T^{\uparrow \lambda})r(S^{\uparrow\lambda})} \ketbra{T}{S}\otimes \ketbra{T_\mu}{S_\mu},
\end{equation}
where $T_\mu$ denotes the standard Young tableau obtained from $T$ by moving the largest box of $T$ to the position that results in shape $\mu$ (such moving always exists since $\mu$ is adjacent to $\nu$). Since the summation is taken over all $\mu\vdash_d k$ that is adjacent to $\nu$, by re-naming $\mu\cap\nu$ to $\tau$, we can write \cref{eq-6161627} in an equivalent form:
\begin{equation}
\bigoplus_{\substack{\nu\vdash_d k}}\frac{1}{\dim(\mathcal{Q}_\nu)}I_{\mathcal{Q}_\nu}\otimes \bigoplus_{\tau:\tau\nearrow\nu} \bigoplus_{\substack{\mu:\tau\nearrow\mu\\ \mu\neq\nu\\ \ell(\mu)\leq d}}\frac{\dim(\mathcal{P}_{\mu\cup \nu})}{\dim(\mathcal{P}_\mu)}I_{\mathcal{Q}_\mu} \otimes \sum_{\substack{T,S\in\Tab(\nu,\tau)}}\frac{\sqrt{(r(T^{\uparrow \mu\cup\nu})^2-1)(r(S^{\uparrow \mu\cup\nu})^2-1)}}{r(T^{\uparrow \mu\cup\nu})r(S^{\uparrow\mu\cup\nu})} \ketbra{T}{S}\otimes \ketbra{T_\mu}{S_\mu}.\label{eq-6150140}
\end{equation}
Note that for $T,S\in\Tab(\nu,\tau)$, $r(T^{\uparrow \mu\cup\nu})=r(S^{\uparrow\mu\cup\nu})=c(\mu\setminus\nu)-c(\nu\setminus\tau)$. 
Then, by \cref{lemma-6161632}, we can write \cref{eq-6150140} as
\begin{equation}\label{eq-6161638}
\frac{k+1}{k}\bigoplus_{\substack{\nu\vdash_d k}}\frac{1}{\dim(\mathcal{Q}_\nu)}I_{\mathcal{Q}_\nu}\otimes \bigoplus_{\tau:\tau\nearrow\nu} \bigoplus_{\substack{\mu:\tau\nearrow\mu\\ \mu\neq\nu}}\frac{\dim(\mathcal{P}_{\nu})}{\dim(\mathcal{P}_{\tau})}I_{\mathcal{Q}_\mu} \otimes \sum_{\substack{T,S\in\Tab(\nu,\tau)}}\ketbra{T}{S}\otimes \ketbra{T_\mu}{S_\mu},
\end{equation}
where we ignore the constraint $\ell(\mu)\leq d$ since $I_{\mathcal{Q}_\mu}=0$ when $\ell(\mu)>d$.

\paragraph{Summands involving the term $\ketbra{T}{S}$.}
Then, consider the term $\ketbra{T}{S}$. We have
\begin{align}
&\tr_{\mathcal{H}_{2k}}\!\left(\frac{\dim(\mathcal{P}_\lambda)}{\dim(\mathcal{P}_\mu)\dim(\mathcal{Q}_\lambda)}I_{\mathcal{Q}_\lambda}\otimes I_{\mathcal{Q}_\mu} \otimes \frac{1}{r(T)r(S)}\ketbra{T}{S}\otimes\ketbra{T^\downarrow}{S^\downarrow}\right) \nonumber\\
=&\frac{\dim(\mathcal{P}_\lambda)}{\dim(\mathcal{P}_\mu)\dim(\mathcal{Q}_{\mu})}I_{\mathcal{Q}_\mu}\otimes I_{\mathcal{Q}_\mu}\otimes \frac{1}{r(T)r(S)}\ketbra{T^\downarrow}{S^{\downarrow}}\otimes\ketbra{T^\downarrow}{S^\downarrow}.\label{eq-6142153}
\end{align}
For those summands of the form in \cref{eq-6142153}, by re-naming the variables $\mu$ to $\nu$, $T^\downarrow$ to $T$ and $S^\downarrow$ to $S$, we can write the corresponding summation as
\begin{equation}
\bigoplus_{\nu\vdash_d k}\frac{1}{\dim(\mathcal{P}_\nu)\dim(\mathcal{Q}_\nu)}I_{\mathcal{Q}_\nu}\otimes I_{\mathcal{Q}_\nu}\otimes \sum_{\substack{\lambda:\nu\nearrow \lambda\\ \ell(\lambda)\leq d}}\,\sum_{\substack{T,S\in\Tab(\nu)}}\frac{\dim(\mathcal{P}_\lambda)}{r(T^{\uparrow \lambda})r(S^{\uparrow \lambda})}\ketbra{T}{S}\otimes\ketbra{T}{S}.\label{eq-6150141}
\end{equation}
Then, note that $\Tab(\nu)=\bigcup_{\tau:\tau\nearrow\nu}\Tab(\nu,\tau)$. We can write \cref{eq-6150141} as
\begin{equation}\label{eq-6161741}
\bigoplus_{\nu\vdash_d k}\frac{1}{\dim(\mathcal{P}_\nu)\dim(\mathcal{Q}_\nu)}I_{\mathcal{Q}_\nu}\otimes I_{\mathcal{Q}_\nu}\otimes \sum_{\substack{\lambda:\nu\nearrow \lambda\\ \ell(\lambda)\leq d}}\,\sum_{\substack{\tau:\tau\nearrow\nu\\\kappa:\kappa\nearrow\nu}}\,\sum_{\substack{T\in\Tab(\nu,\tau)\\ S\in\Tab(\nu,\kappa)}}\frac{\dim(\mathcal{P}_\lambda)}{r(T^{\uparrow \lambda})r(S^{\uparrow \lambda})}\ketbra{T}{S}\otimes\ketbra{T}{S}.
\end{equation}
Note that if we add
\begin{equation}\label{eq-7101659}
\bigoplus_{\nu\vdash_d k}\frac{1}{\dim(\mathcal{P}_\nu)\dim(\mathcal{Q}_\nu)}I_{\mathcal{Q}_\nu}\otimes I_{\mathcal{Q}_\nu}\otimes \sum_{\substack{\lambda:\nu\nearrow \lambda\\ \ell(\lambda)=d+1}}\,\sum_{\substack{\tau:\tau\nearrow\nu\\\kappa:\kappa\nearrow\nu}}\,\sum_{\substack{T\in\Tab(\nu,\tau)\\ S\in\Tab(\nu,\kappa)}}\frac{\dim(\mathcal{P}_\lambda)}{r(T^{\uparrow \lambda})r(S^{\uparrow \lambda})}\ketbra{T}{S}\otimes\ketbra{T}{S}
\end{equation}
to \cref{eq-6161741}, we can remove the constraint $\ell(\lambda)\leq d$ and obtain:
\begin{equation}\label{eq-7101646}
\bigoplus_{\nu\vdash_d k}\frac{1}{\dim(\mathcal{P}_\nu)\dim(\mathcal{Q}_\nu)}I_{\mathcal{Q}_\nu}\otimes I_{\mathcal{Q}_\nu}\otimes \sum_{\substack{\lambda:\nu\nearrow \lambda}}\,\sum_{\substack{\tau:\tau\nearrow\nu\\\kappa:\kappa\nearrow\nu}}\,\sum_{\substack{T\in\Tab(\nu,\tau)\\ S\in\Tab(\nu,\kappa)}}\frac{\dim(\mathcal{P}_\lambda)}{r(T^{\uparrow \lambda})r(S^{\uparrow \lambda})}\ketbra{T}{S}\otimes\ketbra{T}{S}.
\end{equation}
Note that \cref{eq-7101659} is positive semidefinite.
Therefore, \cref{eq-6161741} $\sqsubseteq$ \cref{eq-7101646}.
Note that $r(T^{\uparrow\lambda})=c(\lambda\setminus\nu)-c(\nu\setminus\tau)> 0$ and $r(S^{\uparrow\lambda})=c(\lambda\setminus\nu)-c(\nu\setminus\kappa)> 0$.
Then, in \cref{eq-7101646}, by \cref{lemma-6161659}, we can discard those terms such that $\tau\neq\kappa$. Therefore, \cref{eq-7101646} simplifies to:
\begin{equation}\label{eq-6161754}
\bigoplus_{\nu\vdash_d k}\frac{1}{\dim(\mathcal{P}_\nu)\dim(\mathcal{Q}_\nu)}I_{\mathcal{Q}_\nu}\otimes I_{\mathcal{Q}_\nu}\otimes \sum_{\substack{\lambda:\nu\nearrow \lambda}}\,\sum_{\tau:\tau\nearrow\nu}\,\sum_{\substack{T,S\in\Tab(\nu,\tau)}}\frac{\dim(\mathcal{P}_\lambda)}{r(T^{\uparrow \lambda})r(S^{\uparrow \lambda})}\ketbra{T}{S}\otimes\ketbra{T}{S}.
\end{equation}
Then, by \cref{lemma-6161700}, we can further simplify \cref{eq-6161754} to
\begin{equation}\label{eq-6161806}
\frac{k+1}{k}\cdot \bigoplus_{\nu\vdash_d k}\frac{1}{\dim(\mathcal{Q}_\nu)}I_{\mathcal{Q}_\nu}\otimes I_{\mathcal{Q}_\nu}\otimes \sum_{\tau:\tau\nearrow\nu}\frac{\dim(\mathcal{P}_\nu)}{\dim(\mathcal{P}_\tau)}\sum_{\substack{T,S\in\Tab(\nu,\tau)}}\ketbra{T}{S}\otimes\ketbra{T}{S}.
\end{equation}
In summary, we proved that the sum of the terms involving $\ketbra{T}{S}$ 
(i.e., \cref{eq-6150141}) is less than (w.r.t. L\"owner order) or equal to \cref{eq-6161806}.

\paragraph{Putting it all together.}
Then, combining the above results, we can see that the sum of \cref{eq-6161638} and \cref{eq-6161806} is:
\begin{align}
\frac{k+1}{k}\cdot\bigoplus_{\nu\vdash_d k}\frac{1}{\dim(\mathcal{Q}_\nu)} I_{\mathcal{Q}_\nu}\otimes \bigoplus_{\tau:\tau\nearrow\nu} \bigoplus_{\substack{\mu:\tau\nearrow\mu}}\frac{\dim(\mathcal{P}_\nu)}{\dim(\mathcal{P}_\tau)}I_{\mathcal{Q}_\mu}\otimes\sum_{T,S\in\Tab(\nu,\tau)}\ketbra{T}{S}\otimes\ketbra{T_\mu}{S_\mu}.\label{eq-6170410}
\end{align}
Note that in \cref{eq-6170410}, $\sum_{\substack{\mu:\tau\nearrow\mu}}I_{\mathcal{Q}_\mu}\otimes\ketbra{T_\mu}{S_\mu}$ is an operator on $\bigotimes_{i=1}^k\mathcal{H}_{2i-1}$. Thus \cref{fact-6101953} implies that 
\[\bigoplus_{\substack{\mu:\tau\nearrow\mu}}I_{\mathcal{Q}_\mu}\otimes\ketbra{T_\mu}{S_\mu}=I_{\mathcal{H}_{2k-1}}\otimes I_{\mathcal{Q}_\tau}\otimes\ketbra{T^\downarrow}{S^\downarrow}.\]
This means \cref{eq-6170410} is
\begin{align}
&\frac{k+1}{k} I_{\mathcal{H}_{2k-1}}\otimes \bigoplus_{\nu\vdash_d k}\bigoplus_{\tau:\tau\nearrow\nu}\frac{\dim(\mathcal{P}_\nu)}{\dim(\mathcal{P}_\tau)\dim(\mathcal{Q}_\nu)} I_{\mathcal{Q}_\nu}\otimes  I_{\mathcal{Q}_\tau}\otimes\sum_{T,S\in\Tab(\nu,\tau)}\ketbra{T}{S}\otimes\ketbra{T^\downarrow}{S^\downarrow}\nonumber\\
=&\frac{k+1}{k} \cdot I_{\mathcal{H}_{2k-1}}\otimes A_{k-1}.\nonumber
\end{align}

\subsubsection{The case of \texorpdfstring{$k=1$}{k=1}}
Then, we prove
\[\tr_{\mathcal{H}_2}(A_1)=\frac{2}{d}\cdot I_{\mathcal{H}_1}\otimes I_{\mathcal{H}_{2n+2}}.\]
by explicit calculation. 
Specifically, let $\lambda=\scalebox{.45}{\begin{ytableau}~&~\end{ytableau}}$ and $\mu=\scalebox{.45}{\begin{ytableau}~\\~\end{ytableau}}$ be the two Young diagrams of $2$ boxes and $\nu=\scalebox{.45}{\begin{ytableau}~\end{ytableau}}$ be the Young diagram of $1$ box. Let $T=\scalebox{.45}{\begin{ytableau}1&2\end{ytableau}}$ and $S=\scalebox{.45}{\begin{ytableau}1\\2\end{ytableau}}$. Thus, $\nu\nearrow\lambda$, $\nu\nearrow\mu$, $\dim(\mathcal{P}_\lambda)=\dim(\mathcal{P}_\mu)=\dim(\mathcal{P}_\nu)=1$, $\dim(\mathcal{Q}_\lambda)=\frac{d(d+1)}{2}$, $\dim(\mathcal{Q}_\mu)=\frac{d(d-1)}{2}$ and $\dim(\mathcal{Q}_\nu)=d$.
Then, we have
\begin{align}
A_1&=\frac{2}{d(d+1)} I_{\mathcal{Q}_\lambda}\otimes I_{\mathcal{Q}_\nu}\otimes \ketbra{T}{T}\otimes \ketbra{T^\downarrow}{T^\downarrow}+\frac{2}{d(d-1)}I_{\mathcal{Q}_\mu}\otimes I_{\mathcal{Q}_\nu}\otimes \ketbra{S}{S}\otimes \ketbra{S^\downarrow}{S^\downarrow}.\nonumber
\end{align}
Next, we swap the subsystems $\mathcal{H}_{2}$ and $\mathcal{H}_{2n+2}$. In fact, note that the swap $s_1$ acts trivially on $\ket{T}$ and $\ket{S}$. Thus, we can directly calculate the partial trace on $A_1$ by \cref{lemma-6120245}:
\begin{align}
\tr_{\mathcal{H}_2}(A_1)&=\frac{1}{d}I_{\mathcal{Q}_\nu}\otimes I_{\mathcal{Q}_\nu}\otimes\ketbra{T^\downarrow}{T^\downarrow}\otimes\ketbra{T^\downarrow}{T^\downarrow}+\frac{1}{d}I_{\mathcal{Q}_\nu}\otimes I_{\mathcal{Q}_\nu}\otimes\ketbra{S^\downarrow}{S^\downarrow}\otimes\ketbra{S^\downarrow}{S^\downarrow}\nonumber\\
&=\frac{2}{d}I_{\mathcal{H}_{1}}\otimes I_{\mathcal{H}_{2n+2}},\label{eq-6132124}
\end{align}
where \cref{eq-6132124} is because $T^\downarrow=S^\downarrow=\scalebox{.45}{\begin{ytableau}1\end{ytableau}}$ and thus $\ketbra{T^\downarrow}{T^\downarrow}=\ketbra{S^\downarrow}{S^\downarrow}=I_{\mathcal{P}_\nu}$.

\subsection{Generalized time-reversal}
\label{sub:generalized_time_reversal}

In this section, let us prove \Cref{cor:generalized-time-reversal},
which gives a lower bound for generalized time-reversal of an unknown unitary. Here, we need to use the following Dirichlet's approximation theorem.

\begin{lemma}[Dirichlet's approximation theorem]
    \label{lmm:Dirichlet}
    For any real number $t\in \mathbb{R}$ and integer $N\geq 1$, there exist integers $a,b$ such that $1\leq b\leq N$ and
    $\abs*{bt-a}< 1/N$.
\end{lemma}

Then, we are able to give the proof of \cref{cor:generalized-time-reversal}.

\begin{proof}[Proof of \Cref{cor:generalized-time-reversal}]
    Suppose $\calA$ is an algorithm that implements the unitary $U^{-t}=e^{-iHt}$ to within constant diamond norm error $\epsilon\leq 10^{-5}$ for some constant $t\geq 0.1$, using $n$ queries to $U$, where $H$ is a Hamiltonian satisfying $e^{iH}=e^{i\theta}U$ for some real number $\theta$ and $\norm*{H}\leq \pi$.
    By \Cref{lmm:Dirichlet}, there exist integers $a'$ and $1\leq b'\leq 10^3$ such that
    $\abs*{b't - a'} < 10^{-3}$.
    Let $b=10b',a=10a'$. Then, we have $\abs*{bt-a}\leq 0.01$, where $10\leq b\leq 10^4$. We can also see $a\geq 1$ since $t\geq 0.1$.

    Now let us construct an algorithm to approximately implement $U^{-1}$, based on algorithm $\calA$.
    First, we repeat algorithm $\calA$ $b$ times,
    and then we obtain an algorithm $\calB$ to implement the unitary $e^{-iHbt}$ to within diamond norm error $b\epsilon\leq 0.1$.
    Second, we show the unitary $e^{-iHbt}$ is close to $e^{-iHa}$ in diamond distance, as follows.
    Let $\calV_1(\cdot)=e^{-iHbt}(\cdot)e^{iHbt}$ and $\calV_2(\cdot)=e^{-iHa}(\cdot)e^{iHa}$ be two quantum channels for $e^{-iHbt}$ and $e^{-iHa}$, respectively.
    Using the inequality between the diamond norm and the operator norm (see, e.g., \cite[Proposition 1.6]{haah2023query}) yields
    \begin{equation}
        \label{eq:diamond-operator}
        \norm*{\calV_1-\calV_2}_{\diamond}\leq 2\norm*{e^{-iHbt}-e^{-iHa}},
    \end{equation}
    where $\|\cdot \|$ is the operator norm.
    To further bound the operator norm on the RHS of \Cref{eq:diamond-operator},
    we observe
    \begin{equation*}
        \norm*{e^{-iHbt}-e^{-iHa}}=\norm*{e^{-iHbt}\parens*{I-e^{-iH\parens*{a-bt}}}}=\norm*{I-e^{-iH\parens*{a-bt}}},
    \end{equation*}
    which can be upper bounded by $\abs*{a-bt}\cdot\norm*{H}\leq 0.01\cdot\pi$,
    through looking at the eigenvalues and using $\abs*{1-e^{ix}}\leq \abs*{x}$.
    Consequently, $\norm*{\calV_1-\calV_2}_{\diamond}\leq 0.02\cdot\pi< 0.1$. 
    That is, $\calB$ is also an algorithm to implement the unitary $e^{-iHa}$ to within diamond norm error $b\epsilon + 0.02\cdot\pi < 0.2$.
    
    Finally, if we use $a-1\geq 0$ more queries to $U$ followed by algorithm $\calB$,
    we obtain an algorithm $\calC$ to implement the time-reverse $U^{-1}$ (up to a global phase) to within diamond norm error $0.2$.
    Note that algorithm $\calC$ only uses $bn+a-1$ queries to $U$, and both $a$ and $b$ are constants. 
    Therefore, by \Cref{coro-6290408}, we have $n=\Omega(d^2)$.    
    
\end{proof}

\section{Deferred lemmas}
\label{sec:deferred_lemmas}
\subsection{Raising and lowering of Young diagrams}
In this subsection, we present two representation-theoretic results concerning the raising and lowering of Young diagrams. These results follow from standard techniques involving the Clebsch-Gordan transform and Schur transform~\cite{bacon2007quantum,bacon2006efficient,harrow2005applications} (see also \cite{yoshida2023reversing,studzinski2022efficient}). However, for the reader's convenience, we provide simple and self-contained proofs here.
\begin{lemma}\label{fact-6101953}
Let $n\geq 2$ be an integer and $\mathcal{H}_1, \mathcal{H}_2,\ldots,\mathcal{H}_n$ be a sequence of Hilbert spaces such that $\mathcal{H}_i\cong\mathbb{C}^d$ for $1\leq i\leq n$. Let $\mathfrak{S}_n$ act on $\bigotimes_{i=1}^n\mathcal{H}_i$ and consider the corresponding decomposition $\bigotimes_{i=1}^n\mathcal{H}_i\stackrel{\mathbb{U}_d\times\mathfrak{S}_n}{\cong} \bigoplus_{\substack{\lambda\vdash_d n}} \mathcal{Q}_\lambda\otimes\mathcal{P}_\lambda$. Then, for $\mu\vdash_d n-1$, $T,S\in\Tab(\mu)$, we have
\[I_{\mathcal{H}_n}\otimes I_{\mathcal{Q}_\mu}\otimes \ketbra{T}{S}=\bigoplus_{\lambda:\mu\nearrow \lambda} I_{\mathcal{Q}_\lambda}\otimes \ketbra{T^{\uparrow \lambda}}{S^{\uparrow \lambda}},\]
where $T^{\uparrow \lambda}$ denotes the standard Young tableau obtained from $T$ by adding a box filled with $n$, resulting in the shape $\lambda$,
and we take the convention that $\mathcal{Q}_\lambda=0$ when $\ell(\lambda)>d$.
\end{lemma}
\begin{proof}
We consider the linear operator $I_{\mathcal{Q}_\mu}\otimes \ket{T}\in\mathcal{L}(\mathcal{Q}_\mu,\mathcal{Q}_\mu\otimes\mathcal{P}_\mu)\subseteq \mathcal{L}(\mathcal{Q}_\mu,\bigotimes_{i=1}^{n-1}\mathcal{H}_{i})$. By Pieri's rule:
\[\mathbb{C}^d\otimes \mathcal{Q}_\mu\stackrel{\mathbb{U}_{d}}{\cong}\bigoplus_{\lambda:\mu\nearrow\lambda}\mathcal{Q}_{\lambda},\]
we can see that
\begin{equation}\label{eq-6101547}
I_{\mathcal{H}_{n}}\otimes I_{\mathcal{Q}_\mu}\otimes \ket{T}=\bigoplus_{\lambda:\mu\nearrow \lambda} I_{\mathcal{Q}_\lambda}\otimes \ket{T_\lambda},
\end{equation}
where $\ket{T_\lambda}$ is a unit vector and by Schur-Weyl duality on $\bigotimes_{i=1}^{n}\mathcal{H}_{i}$, we know that $\ket{T_\lambda}\in\mathcal{P}_\lambda$. 
Suppose $T=\mu^{(1)}\rightarrow\cdots\rightarrow\mu^{(n-1)}$ where $\mu^{(n-1)}=\mu$. For an integer $1\leq k\leq n-1$, we apply $P(e_{\mu^{(k)}})$ (i.e., the action of $e_{\mu^{(k)}}$ on $\bigotimes_{i=1}^{n}\mathcal{H}_{i}$, where $e_{\mu^{(k)}}$ is defined in \cref{eq-6122124}) to both sides of \cref{eq-6101547}. 
For the LHS, since $e_{\mu^{(k)}}\in\mathbb{C}\mathfrak{S}_k\subseteq\mathbb{C}\mathfrak{S}_{n-1}$, $P(e_{\mu^{(k)}})$ acts non-trivially only on $\bigotimes_{i=1}^{n-1}\mathcal{H}_{i}$, thus we have
\[P(e_{\mu^{(k)}})(I_{\mathcal{H}_{n}}\otimes I_{\mathcal{Q}_\mu}\otimes \ket{T})=I_{\mathcal{H}_{n}}\otimes I_{\mathcal{Q}_\mu}\otimes e_{\mu^{(k)}}\ket{T}=I_{\mathcal{H}_{n}}\otimes I_{\mathcal{Q}_\mu}\otimes \ket{T}.\] 
For the RHS, we have
\[P(e_{\mu^{(k)}})\left(\bigoplus_{\lambda:\mu\nearrow\lambda}I_{\mathcal{Q}_\lambda}\otimes \ket{T_\lambda}\right)=\bigoplus_{\lambda:\mu\nearrow \lambda} I_{\mathcal{Q}_\lambda}\otimes e_{\mu^{(k)}}\ket{T_\lambda}.\]
Therefore, we can see that
\[e_{\mu^{(k)}}\ket{T_\lambda}=\ket{T_\lambda},\]
for any $1\leq k\leq n-1$.
This means $\ket{T_\lambda}\in\mathcal{P}_\lambda$, when restricted as a vector in a representation of $\mathfrak{S}_k$, is a vector in $\mathcal{P}_{\mu^{(k)}}$. Therefore, $\ket{T_\lambda}=\ket{\mu^{(1)}\rightarrow\cdots\rightarrow\mu^{(n-1)}\rightarrow \lambda}=\ket{T^{\uparrow\lambda}}$.

We can analogously prove $I_{\mathcal{H}_{n}}\otimes I_{\mathcal{Q}_\mu}\otimes \bra{S}=\bigoplus_{\lambda:\mu\nearrow \lambda} I_{\mathcal{Q}_\lambda}\otimes \bra{S^{\uparrow\lambda}}$. This completes the proof.
\end{proof}

The following result can be viewed, in some sense, as a ``dual'' of \cref{fact-6101953}.
\begin{lemma}\label{lemma-6120245}
Let $n\geq 2$ be an integer and $\mathcal{H}_1, \mathcal{H}_2,\ldots,\mathcal{H}_n$ be a sequence of Hilbert spaces such that $\mathcal{H}_i\cong\mathbb{C}^d$ for $1\leq i\leq n$. Let $\mathfrak{S}_n$ act on $\bigotimes_{i=1}^n\mathcal{H}_i$ and consider the corresponding decomposition $\bigotimes_{i=1}^n\mathcal{H}_i\stackrel{\mathbb{U}_d\times\mathfrak{S}_n}{\cong} \bigoplus_{\lambda\vdash_d n} \mathcal{Q}_\lambda\otimes\mathcal{P}_\lambda$. Then, for $\lambda\vdash_d n$, $T,S\in\Tab(\lambda)$, we have
\[\tr_{\mathcal{H}_n}(I_{\mathcal{Q}_{\lambda}}\otimes \ketbra{T}{S})=\mathbbm{1}_{\Sh(T^\downarrow)=\Sh(S^\downarrow)}\cdot\frac{\dim(\mathcal{Q}_\lambda)}{\dim(\mathcal{Q}_{\Sh(T^\downarrow)})}\cdot I_{\mathcal{Q}_{\Sh(T^\downarrow)}}\otimes \ketbra{T^\downarrow}{S^\downarrow},\]
where $T^\downarrow$ denotes the standard Young tableau obtained from $T$ by removing the box containing the largest integer.
\end{lemma}
\begin{proof}
Let 
\begin{equation}\label{eq-6130308}
W\coloneqq \tr_{\mathcal{H}_n}(I_{\mathcal{Q}_\lambda}\otimes \ketbra{T}{S}).
\end{equation}
Note that $W\in\mathcal{L}(\bigotimes_{i=1}^{n-1}\mathcal{H}_i)$. For any $U\in\mathbb{U}_d$, we have
\begin{align}
U^{\otimes n-1} W U^{\dag\otimes n-1}&= \tr_{\mathcal{H}_n}\!\left(U^{\otimes n} (I_{\mathcal{Q}_\lambda}\otimes \ketbra{T}{S})U^{\dag\otimes n}\right)\nonumber\\
&=\tr_{\mathcal{H}_n}\!\left(I_{\mathcal{Q}_\lambda}\otimes\ketbra{T}{S}\right)\nonumber\\
&=\tr_{\mathcal{H}_n}(I_{\mathcal{Q}_\lambda}\otimes \ketbra{T}{S}).\nonumber
\end{align}
By Schur-Weyl duality on $\bigotimes_{i=1}^{n-1}\mathcal{H}_i$ and Schur's lemma, this means
\begin{equation}\label{eq-6122217}
W=\bigoplus_{\mu\vdash_d n-1} I_{\mathcal{Q}_\mu}\otimes X_{\mu},
\end{equation}
for some $X_\mu\in\mathcal{L}(\mathcal{P}_\mu)$. 
On the other hand, suppose $T=\lambda_T^{(1)}\rightarrow\lambda_T^{(2)}\rightarrow\cdots\rightarrow \lambda_{T}^{(n)}$ 
and
$S=\lambda_S^{(1)}\rightarrow\lambda_S^{(2)}\rightarrow\cdots\rightarrow \lambda_{S}^{(n)}$, where $\lambda_T^{(n)}=\lambda_S^{(n)}=\lambda$. 
For $k\leq n-1$ and $\nu,\tau\vdash k$, we left-apply $P(e_\nu)$ (i.e., the action of $e_\nu$ on $\bigotimes_{i=1}^{n-1}\mathcal{H}_{i}$, where $e_\nu$ is defined in \cref{eq-6122124}) and right-apply $P(e_\tau)$ to $W$. 
First, by the definition of $W$ (c.f. \cref{eq-6130308}), we have
\begin{align}
P(e_\nu)W P(e_\tau)&=\tr_{\mathcal{H}_n}\!\left(P(e_\nu)(I_{\mathcal{Q}_\lambda}\otimes  \ketbra{T}{S}) P(e_\tau)\right)\label{eq-6122227}\\
&=\tr_{\mathcal{H}_n}(I_{\mathcal{Q}_\lambda}\otimes e_\nu \ketbra{T}{S} e_\tau)\nonumber\\
&=\mathbbm{1}_{\nu=\lambda_T^{(k)}}\cdot \mathbbm{1}_{\tau=\lambda_{S}^{(k)}}\cdot \tr_{\mathcal{H}_n}(I_{\mathcal{Q}_\lambda}\otimes \ketbra{T}{S})\nonumber \\
&=\mathbbm{1}_{\nu=\lambda_T^{(k)}}\cdot \mathbbm{1}_{\tau=\lambda_{S}^{(k)}}\cdot W,\label{eq-6122254}
\end{align}
where \cref{eq-6122227} is because $e_\nu,e_\tau\in\mathbb{C}\mathfrak{S}_k\subseteq \mathbb{C}\mathfrak{S}_{n-1}$ acts non-trivially only on $\bigotimes_{i=1}^{n-1}\mathcal{H}_i$.
Then, by recursively using \cref{eq-6122254}, we know that
\[W= P\!\left(e_{\lambda_T^{(1)}}\right)\cdots P\!\left(e_{\lambda_T^{(n-1)}}\right)\cdot W\cdot P\!\left(e_{\lambda_S^{(n-1)}}\right)\cdots P\!\left(e_{\lambda_S^{(1)}}\right).\]
On the other hand, by using \cref{eq-6122217}, we have
\begin{align}
W&=P\!\left(e_{\lambda_T^{(1)}}\right)\cdots P\!\left(e_{\lambda_T^{(n-1)}}\right)\cdot W\cdot P\!\left(e_{\lambda_S^{(n-1)}}\right)\cdots P\!\left(e_{\lambda_S^{(1)}}\right)\nonumber\\
&=\bigoplus_{\mu\vdash_d n-1} I_{\mathcal{Q}_\mu} \otimes \left(e_{\lambda_T^{(1)}} \cdots e_{\lambda_T^{(n-1)}} \cdot X_\mu \cdot e_{\lambda_S^{(n-1)}} \cdots e_{\lambda_S^{(1)}}\right)\nonumber\\
&= \bigoplus_{\mu\vdash_d n-1} I_{\mathcal{Q}_\mu} \otimes \left(\ketbra{T^\downarrow}{T^\downarrow} \cdot X_\mu \cdot \ketbra{S^\downarrow}{S^\downarrow}\right),\label{eq-6130139}
\end{align}
where \cref{eq-6130139} is by using \cref{eq-61301399}. Note that $\ketbra{T^\downarrow}{T^\downarrow}$ is a linear projector on $\mathcal{P}_{\Sh(T^\downarrow)}$ and $\ketbra{S^\downarrow}{S^\downarrow}$ is a linear projector on $\mathcal{P}_{\Sh(S^\downarrow)}$. This means $\ketbra{T^\downarrow}{T^\downarrow}\cdot X_\mu \cdot \ketbra{S^\downarrow}{S^\downarrow}$ is non-zero only when $\mu=\Sh(T^\downarrow)=\Sh(S^\downarrow)$. Therefore, we have
\begin{equation}\label{eq-6130154}
W=\mathbbm{1}_{\Sh(T^\downarrow)=\Sh(S^\downarrow)}\cdot c\cdot I_{\mathcal{Q}_{\Sh(T^\downarrow)}}\otimes \ketbra{T^\downarrow}{S^\downarrow},
\end{equation}
for some number $c$. 

Now suppose $\Sh(T^\downarrow)=\Sh(S^\downarrow)$. Thus both $\ket{T^\downarrow}$ and $\ket{S^\downarrow}$ are in the same irreducible representation $\mathcal{P}_{\Sh(T^\downarrow)}$. Let $K\in\mathbb{C}\mathfrak{S}_{n-1}$ such that $K \ket{T^\downarrow}=\ket{S^\downarrow}$ (such $K$ always exists by the Jacobson density theorem~\cite{etingof2011introduction}). Then, $K$ also satisfies $K\ket{T}=\ket{S}$. To see this, we view  $\ket{T}$ and $\ket{S}$ as vectors in the restricted representation $\textup{Res}^{\mathfrak{S}_n}_{\mathfrak{S}_{n-1}}\mathcal{P}_\lambda$. By the definition of Young basis, they both are in the same irreducible representation $\mathcal{P}_{\Sh(T^\downarrow)}$ and correspond to $\ket{T^\downarrow}$ and $\ket{S^\downarrow}$, respectively. Therefore, by the definition of $K$, $K$ maps $\ket{T}$ to $\ket{S}$.
Then, consider the trace $\tr(P(K)\cdot W)$, where $P(K)$ is the action of $K$ on $\bigotimes_{i=1}^{n-1}\mathcal{H}_i$. On the one hand, by \cref{eq-6130154}, we have
\begin{align}
\tr(P(K)\cdot W)&=\tr\!\left(c\cdot I_{\mathcal{Q}_{\Sh(T^\downarrow)}}\otimes K\ketbra{T^\downarrow}{S^\downarrow}\right)=\tr\!\left(c\cdot I_{\mathcal{Q}_{\Sh(T^\downarrow)}}\otimes \ketbra{S^\downarrow}{S^\downarrow}\right)= c\cdot\dim(\mathcal{Q}_{\Sh(T^\downarrow)}).\nonumber
\end{align}
On the other hand, by \cref{eq-6130308}, we have
\begin{align}
\tr(P(K)\cdot W)&=\tr\!\left(\tr_{\mathcal{H}_n}(I_{\mathcal{Q}_\lambda}\otimes K\ketbra{T}{S})\right)=\tr\!\left(\tr_{\mathcal{H}_n}(I_{\mathcal{Q}_\lambda}\otimes \ketbra{S}{S})\right)=\dim(\mathcal{Q}_\lambda).\nonumber
\end{align}
Therefore, $c=\dim(\mathcal{Q}_\lambda)/\dim(\mathcal{Q}_{\Sh(T^\downarrow)})$, completing the proof upon substitution into \cref{eq-6130154}.
\end{proof}

\subsection{Combinatorics on Young diagrams}\label{sec-722129}
Here, we introduce some combinatorial results on Young diagrams.
The first result is by a direct calculation using the hook length formula.

\begin{lemma}\label{lemma-6161632}
Suppose $\mu,\nu\vdash n$, $\mu\neq\nu$ and $\mu,\nu$ are adjacent. Then
\[\frac{\dim(\mathcal{P}_{\mu})\dim(\mathcal{P}_\nu)}{\dim(\mathcal{P}_{\mu\cup\nu})\dim(\mathcal{P}_{\mu\cap\nu})}= \frac{n}{n+1}\left(1-\frac{1}{(c(\mu\setminus \nu)-c(\nu\setminus\mu))^2}\right),\]
where $\mu\setminus\nu$ contains only one box and $c(\mu\setminus \nu)$ is the axial coordinate of this box (and similarly for $\nu\setminus\mu$).
\end{lemma}
\begin{proof}
Let $\square_\mu$ be the box in $\mu\setminus\nu$ and $\square_\nu$ be the box in $\nu\setminus\mu$. Define $\tau\coloneqq \mu\cap\nu$. Denote by $H_\lambda(\square)$ the set of boxes in $\lambda$ that are directly to the left and above the box $\square$, including $\square$ itself (i.e., a hook in the reverse direction). 
By the hook length formula (see \cref{eq-6192141}), we can easily see that
\[\frac{\dim(\mathcal{P}_\mu)}{\dim(\mathcal{P}_{\tau})}=n\cdot\prod_{\square\in H_\tau(\square_{\mu})}\frac{h_{\tau}(\square)}{h_\tau(\square)+1}.\]
Similarly
\[\frac{\dim(\mathcal{P}_\nu)}{\dim(\mathcal{P}_{\tau})}=n\cdot\prod_{\square\in H_\tau(\square_{\nu})}\frac{h_{\tau}(\square)}{h_\tau(\square)+1}.\]
Then, note that $\square_\mu,\square_\nu$ must be at different rows and columns. Thus, there is exactly one box in $H_\tau(\square_\mu)\cap H_\tau(\square_\nu)$, which we denote by $\widetilde{\square}$. Then, we have
\[\frac{\dim(\mathcal{P}_{\mu\cup\nu})}{\dim(\mathcal{P}_\tau)}=n(n+1)\cdot\left(\prod_{\substack{\square\in H_\tau(\square_{\mu})\\\square\neq \widetilde{\square}}}\frac{h_{\tau}(\square)}{h_\tau(\square)+1}\right)\cdot \left(\prod_{\substack{\square\in H_\tau(\square_{\nu})\\\square\neq \widetilde{\square}}}\frac{h_{\tau}(\square)}{h_\tau(\square)+1}\right)\cdot \left(\frac{h_\tau(\widetilde{\square})}{h_\tau(\widetilde{\square})+2}\right).\]
Then, we have
\begin{align}
\frac{\dim(\mathcal{P}_{\mu})\dim(\mathcal{P}_\nu)}{\dim(\mathcal{P}_{\mu\cup\nu})\dim(\mathcal{P}_{\tau})}&=\frac{n}{n+1}\cdot \left(\frac{h_\tau(\widetilde{\square})}{h_\tau(\widetilde{\square})+1}\right)^2\cdot \frac{h_\tau(\widetilde{\square})+2}{h_\tau(\widetilde{\square})}\nonumber\\
&=\frac{n}{n+1}\cdot \left(1-\frac{1}{(h_\tau(\widetilde{\square})+1)^2}\right)\nonumber\\
&=\frac{n}{n+1}\left(1-\frac{1}{(c(\square_\mu)-c(\square_\nu))^2}\right),\label{eq-6200014}
\end{align}
where \cref{eq-6200014} is because $h_\tau(\widetilde{\square})+1=|c(\square_\mu)-c(\square_\nu)|$ (see, e.g., \cref{fig-6192328}).
\begin{figure}[ht]
    \centering
    \begin{subfigure}[b]{0.4\linewidth}
    \centering
    \includegraphics[width=1.0\linewidth]{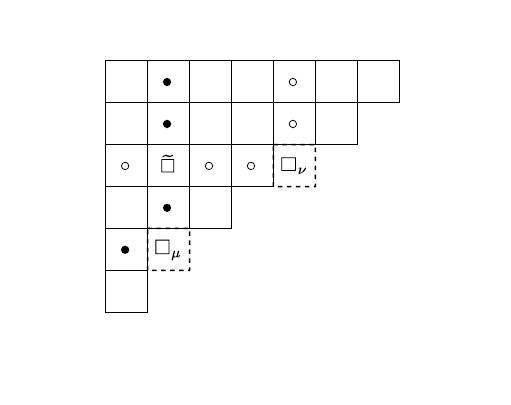}
    \vspace{-6mm}
    \caption{The sets $H_\tau(\square_\mu)$ and $H_\tau(\square_\nu)$.}\label{fig-6192327}
    \end{subfigure}
    \hspace{-0mm}
    \begin{subfigure}[b]{0.4\linewidth}
    \centering
    \includegraphics[width=1.0\linewidth]{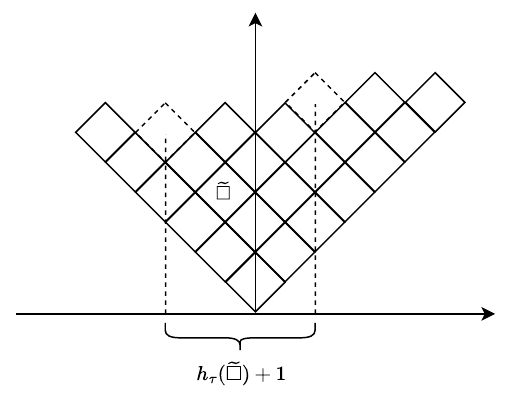}
    \vspace{-6mm}
    \caption{The axial distance between $\square_\mu$ and $\square_\nu$.}\label{fig-6192328}
    \end{subfigure}
    \vspace{-0mm}
    \caption{An example illustrating the idea in the proof of \cref{lemma-6161632}.}
\end{figure}
\end{proof}

Then, we introduce the following results.
While \cref{lemma-6161659} and \cref{lemma-6161700} are known in the literature (see, e.g., \cite{kosuda2003new}), we provide different elementary proofs here for completeness.
Our proofs are based on Kerov's interlacing sequences~\cite{kerov1993transition}.

\begin{lemma}\label{lemma-6161659}
Suppose $\nu\vdash n$, $\tau\nearrow \nu$, $\kappa \nearrow \nu$, $\tau\neq\kappa$. Then,
\[\sum_{\lambda:\nu\nearrow\lambda}\frac{\dim(\mathcal{P}_\lambda)}{\dim(\mathcal{P}_\nu)}\frac{1}{c(\lambda\setminus\nu)-c(\nu\setminus\tau)}\frac{1}{c(\lambda\setminus \nu)-c(\nu\setminus\kappa)}=0,\]
where $\lambda\setminus\nu$ contains only one box and $c(\lambda\setminus \nu)$ is the axial coordinate of this box (and similarly for $\nu\setminus\tau$ and $\nu\setminus\kappa$).
\end{lemma}
\begin{proof}
Note that
\begin{align}
&\sum_{\lambda:\nu\nearrow\lambda}\frac{\dim(\mathcal{P}_\lambda)}{\dim(\mathcal{P}_\nu)}\frac{1}{c(\lambda\setminus\nu)-c(\nu\setminus\tau)}\frac{1}{c(\lambda\setminus \nu)-c(\nu\setminus\kappa)}\nonumber\\
=&\frac{1}{c(\nu\setminus\tau)-c(\nu\setminus\kappa)}\sum_{\lambda:\nu\nearrow\lambda}\frac{\dim(\mathcal{P}_\lambda)}{\dim(\mathcal{P}_\nu)}\left(\frac{1}{c(\lambda\setminus\nu)-c(\nu\setminus\tau)}-\frac{1}{c(\lambda\setminus \nu)-c(\nu\setminus\kappa)}\right)\nonumber
\end{align}
Therefore, it suffices to prove that for any $\tau\nearrow\nu$, we have
\[\sum_{\lambda:\nu\nearrow\lambda}\frac{\dim(\mathcal{P}_\lambda)}{\dim(\mathcal{P}_\nu)}\frac{1}{c(\lambda\setminus\nu)-c(\nu\setminus\tau)}=0.\]

Let $\bm{\alpha}=(\alpha_1,\ldots,\alpha_L)$ and $\bm{\beta}=(\beta_1,\ldots,\beta_{L-1})$ be the interlacing sequences of $\nu$. 
The box in $\nu\setminus \tau$ is a removable box of $\nu$. Thus this box is at $\beta_m$ for some $m$ such that $\beta_m=c(\nu\setminus\tau)$. 
Any $\lambda$ such that $\nu\nearrow\lambda$ is obtained by adding a box at an addable position of $\nu$, which corresponds to an $\alpha_{i}$ for $1\leq i\leq L$. Furthermore, for the $\lambda$ obtained by adding a box at $\alpha_i$, we know $c(\lambda\setminus \nu)=\alpha_i$.
Then, by \cref{eq-6192042}, we have
\[\sum_{\lambda:\nu\nearrow\lambda}\frac{\dim(\mathcal{P}_\lambda)}{\dim(\mathcal{P}_\nu)}\frac{1}{c(\lambda\setminus\nu)-c(\nu\setminus\tau)}=(n+1)\cdot \sum_{i=1}^L \frac{1}{\alpha_i-\beta_m} \prod_{j=1}^{i-1}\frac{\alpha_i-\beta_j}{\alpha_i-\alpha_j} \prod_{j=i+1}^L \frac{\alpha_i-\beta_{j-1}}{\alpha_i-\alpha_j},\]
which is $0$ by \cref{lemma-6181616}.
\end{proof}

\begin{lemma}\label{lemma-6161700}
Suppose $\nu\vdash n$, $\tau\nearrow \nu$. Then
\begin{equation}\label{eq-6200403}
\sum_{\lambda:\nu\nearrow\lambda}\frac{\dim(\mathcal{P}_\lambda)}{\dim(\mathcal{P}_\nu)}\frac{1}{(c(\lambda\setminus\nu)-c(\nu\setminus\tau))^2}=\frac{n+1}{n}\cdot\frac{\dim(\mathcal{P}_\nu)}{\dim(\mathcal{P}_\tau)},
\end{equation}
where $\lambda\setminus\nu$ contains only one box and $c(\lambda\setminus \nu)$ is the axial coordinate of this box (and similarly for $\nu\setminus\tau$).
\end{lemma}
\begin{proof}
Let $\bm{\alpha}=(\alpha_1,\ldots,\alpha_L)$ and $\bm{\beta}=(\beta_1,\ldots,\beta_{L-1})$ be the interlacing sequences of $\nu$. 
The box in $\nu\setminus \tau$ is a removable box of $\nu$. Thus this box is at $\beta_m$ for some $m$ such that $\beta_m=c(\nu\setminus\tau)$. 
Any $\lambda$ such that $\nu\nearrow\lambda$ is obtained by adding a box at an addable position of $\nu$, which corresponds to an $\alpha_{i}$ for $1\leq i\leq L$. Furthermore, for the $\lambda$ obtained by adding a box at $\alpha_i$, we know $c(\lambda\setminus \nu)=\alpha_i$.
Then, by \cref{eq-6192042} (applied to $\frac{\dim(\mathcal{P}_\lambda)}{\dim(\mathcal{P}_\nu)}$) and \cref{eq-6200309} (applied to $\frac{\dim(\mathcal{P}_\tau)}{\dim(\mathcal{P}_\nu)}$), \cref{eq-6200403} is equivalent to
\begin{align}
\sum_{i=1}^L \frac{1}{(\alpha_i-\beta_m)^2}\prod_{j=1}^{i-1}\frac{\alpha_i-\beta_j}{\alpha_i-\alpha_j}\prod_{j=i+1}^L\frac{\alpha_{i}-\beta_{j-1}}{\alpha_i-\alpha_j}
=\frac{1}{(\alpha_L-\beta_m)(\beta_m-\alpha_1)}\prod_{j=1}^{m-1}\frac{\beta_m-\beta_j}{\beta_m-\alpha_{j+1}}\prod_{j=m+1}^{L-1}\frac{\beta_m-\beta_j}{\beta_m-\alpha_j}.\label{eq-6200406}
\end{align}
Then, we prove \cref{eq-6200406} by induction on $L$. For the case $L=2$, $m$ must be $1$, and the holding of \cref{eq-6200406} can be checked by direct calculation. 
Now, suppose \cref{eq-6200406} holds for $L$. We want to prove it also holds for $L+1$, i.e., the following holds:
\begin{align}
\sum_{i=1}^{L+1} \frac{1}{(\alpha_i-\beta_m)^2}\prod_{j=1}^{i-1}\frac{\alpha_i-\beta_j}{\alpha_i-\alpha_j}\prod_{j=i+1}^{L+1}\frac{\alpha_{i}-\beta_{j-1}}{\alpha_i-\alpha_j}
=\frac{1}{(\alpha_{L+1}-\beta_m)(\beta_m-\alpha_1)}\prod_{j=1}^{m-1}\frac{\beta_m-\beta_j}{\beta_m-\alpha_{j+1}}\prod_{j=m+1}^{L}\frac{\beta_m-\beta_j}{\beta_m-\alpha_j}.\label{eq-6200407}
\end{align}

For the case of $L+1$, we can assume without loss of generality that $m\leq L-1$. This is because if $m=L$, then we consider the reversed interlacing sequences $\bm{\alpha}'=(\alpha'_1,\ldots,\alpha'_{L+1})$, $\bm{\beta}'= (\beta'_1,\ldots,\beta'_{L})$ such that $\alpha'_i=-\alpha_{L+2-i}$ and $\beta'_i=-\beta_{L+1-i}$. Then, it is easy to see that the holding of \cref{eq-6200407} on $\bm{\alpha},\bm{\beta}$ with $m=L$ is equivalent to that on $\bm{\alpha}'$, $\bm{\beta}'$ with $m=1$.

It is easy to see that
\[\textup{RHS of \cref{eq-6200407}} - \left(\frac{\beta_m-\beta_L}{\beta_m-\alpha_{L+1}}\times \textup{RHS of \cref{eq-6200406}} \right)=0.\]
By the induction hypothesis, it suffices to prove
\begin{equation}\label{eq-6202104}
 \textup{LHS of \cref{eq-6200407}} - \left(\frac{\beta_m-\beta_L}{\beta_m-\alpha_{L+1}}\times \textup{LHS of \cref{eq-6200406}}\right)=0.
\end{equation}
Then, note that the LHS of \cref{eq-6200407} can be written as
\begin{equation*}
\sum_{i=1}^L \frac{\alpha_i-\beta_L}{(\alpha_i-\beta_m)^2\cdot (\alpha_i-\alpha_{L+1})}\prod_{j=1}^{i-1}\frac{\alpha_i-\beta_j}{\alpha_i-\alpha_j}\prod_{j=i+1}^L\frac{\alpha_{i}-\beta_{j-1}}{\alpha_i-\alpha_j}+ 
\frac{1}{(\alpha_{L+1}-\beta_m)^2}\prod_{j=1}^{L}\frac{\alpha_{L+1}-\beta_j}{\alpha_{L+1}-\alpha_j}.
\end{equation*}
Therefore, the LHS of \cref{eq-6202104} is:
\begin{align}
&\sum_{i=1}^L \frac{\beta_L-\alpha_{L+1}}{(\alpha_i-\beta_m)(\beta_m-\alpha_{L+1})(\alpha_i-\alpha_{L+1})}\prod_{j=1}^{i-1}\frac{\alpha_i-\beta_j}{\alpha_i-\alpha_j}\prod_{j=i+1}^L\frac{\alpha_{i}-\beta_{j-1}}{\alpha_i-\alpha_j}+ 
\frac{1}{(\alpha_{L+1}-\beta_m)^2}\prod_{j=1}^{L}\frac{\alpha_{L+1}-\beta_j}{\alpha_{L+1}-\alpha_j}\nonumber\\
=&\frac{\beta_L-\alpha_{L+1}}{(\beta_m-\alpha_{L+1})^2} \Bigg(\sum_{i=1}^L \left(\frac{1}{\alpha_i-\beta_m}-\frac{1}{\alpha_i-\alpha_{L+1}}\right) \prod_{j=1}^{i-1}\frac{\alpha_i-\beta_j}{\alpha_i-\alpha_j}\prod_{j=i+1}^L\frac{\alpha_{i}-\beta_{j-1}}{\alpha_i-\alpha_j} 
-\frac{1}{\alpha_{L+1}-\alpha_{L}} \prod_{j=1}^{L-1}\frac{\alpha_{L+1}-\beta_j}{\alpha_{L+1}-\alpha_j}\Bigg)\nonumber\\
=&-\frac{\beta_L-\alpha_{L+1}}{(\beta_m-\alpha_{L+1})^2} \Bigg(\sum_{i=1}^L \frac{1}{\alpha_i-\alpha_{L+1}} \prod_{j=1}^{i-1}\frac{\alpha_i-\beta_j}{\alpha_i-\alpha_j}\prod_{j=i+1}^L\frac{\alpha_{i}-\beta_{j-1}}{\alpha_i-\alpha_j} 
+\frac{1}{\alpha_{L+1}-\alpha_{L}} \prod_{j=1}^{L-1}\frac{\alpha_{L+1}-\beta_j}{\alpha_{L+1}-\alpha_j}\Bigg)\label{eq-6202148}\\
=& -\frac{\beta_L-\alpha_{L+1}}{(\beta_m-\alpha_{L+1})^2} \Bigg(\sum_{i=1}^{L+1}\frac{1}{\alpha_i-\beta_{L}}\prod_{j=1}^{i-1}\frac{\alpha_i-\beta_j}{\alpha_i-\alpha_j}\prod_{j=i+1}^{L+1}\frac{\alpha_i-\beta_{j-1}}{\alpha_i-\alpha_j}\Bigg)\nonumber\\
=&\,0,\label{eq-6202149}
\end{align}
where \cref{eq-6202148} is by \cref{lemma-6181616} and \cref{eq-6202149} is also by \cref{lemma-6181616}.
\end{proof}

\subsection{Auxiliary lemmas}
The following \cref{lemma-6181616} is repeatedly used in the proofs of \cref{lemma-6161659} and \cref{lemma-6161700}.
\begin{lemma}\label{lemma-6181616}
Suppose $\alpha_1,\ldots,\alpha_L,\beta_1,\ldots,\beta_{L-1}$ are pairwise distinct. Let $1\leq m\leq L-1$. Then we have
\begin{equation}\label{eq-6200210}
\sum_{i=1}^L \frac{1}{\alpha_i-\beta_m} \prod_{j=1}^{i-1}\frac{\alpha_i-\beta_j}{\alpha_i-\alpha_j} \prod_{j=i+1}^L \frac{\alpha_i-\beta_{j-1}}{\alpha_i-\alpha_j}=0.
\end{equation}
\end{lemma}
\begin{proof}
Note that the LHS of \cref{eq-6200210} can be written as
\[\sum_{i=1}^L \prod_{\substack{j=1\\j\neq m}}^{L-1}(\alpha_i-\beta_j)\prod_{\substack{j=1\\j\neq i}}^{L}\frac{1}{\alpha_i-\alpha_j}.\]
Then we multiply $\prod_{j=2}^{L}(\alpha_1-\alpha_j)$ on it and obtain
\begin{equation}\label{eq-6200238}
\prod_{\substack{j=1\\j\neq m}}^{L-1}(\alpha_1-\beta_j)-\sum_{i=2}^L \prod_{\substack{j=2\\j\neq i}}^L (\alpha_1-\alpha_j)\cdot  \prod_{\substack{j=1\\j\neq m}}^{L-1}(\alpha_i-\beta_j)\prod_{\substack{j=2\\j\neq i}}^L\frac{1}{\alpha_i-\alpha_j}.
\end{equation}
Then, we consider it as a polynomial in variable $\alpha_1$ with parameters $\alpha_2,\ldots,\alpha_{L},\beta_1,\ldots,\beta_{L-1}$. This polynomial is of degree $L-2$. Then, if we set $\alpha_1=\alpha_k$ for a $2\leq k\leq L$, we can see that \cref{eq-6200238} becomes
\begin{align}
&\prod_{\substack{j=1\\j\neq m}}^{L-1}(\alpha_k-\beta_j)- \prod_{\substack{j=2\\j\neq k}}^L (\alpha_k-\alpha_j)\cdot  \prod_{\substack{j=1\\j\neq m}}^{L-1}(\alpha_k-\beta_j)\prod_{\substack{j=2\\j\neq k}}^L\frac{1}{\alpha_k-\alpha_j}\nonumber\\
=&\prod_{\substack{j=1\\j\neq m}}^{L-1}(\alpha_k-\beta_j)-\prod_{\substack{j=1\\j\neq m}}^{L-1}(\alpha_k-\beta_j)\nonumber\\
=&\,0.\nonumber
\end{align}
Therefore, this polynomial has $L-1$ roots $\alpha_2,\ldots,\alpha_L$, which means it must be the $0$ polynomial.
\end{proof}

The following \cref{prop-5121539} provides a representation-theoretic expression of the Haar moment.
This result was also used in the study of entanglement and bipartite quantum state (see, e.g., \cite{matsumoto2007universal,chen2024local}).
\begin{lemma}\label{prop-5121539}
Let $C_U$ be that defined in \cref{def-690411}. We have
\begin{align}
\E_{U}[C_{U}]
&=\bigoplus_{\lambda \vdash_d n+1} \frac{1}{\dim(\mathcal{Q}_\lambda)} I_{\mathcal{Q}_\lambda}\otimes I_{\mathcal{Q}_\lambda}\otimes \kettbbra{I_{\mathcal{P}_\lambda}}{I_{\mathcal{P}_\lambda}}.\nonumber
\end{align}
\end{lemma}
\begin{proof}
Note that $C_{U}=\kettbbra{U}{U}^{\otimes n+1}$ is an $(n+1)$-comb on $(\mathcal{H}_1,\mathcal{H}_2,\ldots,\mathcal{H}_{2n+2})$. We can easily check the following facts.
\begin{fact}\label{fact-5121815}
$\E_U[C_{U}]$ commutes with the action of the group $\mathbb{U}_d\times \mathbb{U}_d$, i.e., for $(V,W)\in\mathbb{U}_d\times\mathbb{U}_d$:
\[(V\otimes W)^{\otimes n+1} \E_U[C_{U}]=\E_U[C_{U}](V\otimes W)^{\otimes n+1}.\]
\end{fact}
\begin{proof}[Proof of \cref{fact-5121815}]
\begin{align}
(V\otimes W)^{\otimes n+1} \E_U[C_{U}] &=\E_U\!\left[(V\otimes W\kettbbra{U}{U})^{\otimes n+1}\right]\nonumber\\
&=\E_U\!\left[\kettbbra{VU W^{\textup{T}}}{U}^{\otimes n+1}\right]\nonumber\\
&=\E_{U'}\!\left[\kettbbra{U'}{V^\dag U' W^*}^{\otimes n+1}\right]\label{eq-692323}\\
&=\E_{U'}\!\left[(\kettbbra{U'}{U'} V\otimes W)^{\otimes n+1}\right]\nonumber\\
&=\E_U[C_{U}](V\otimes W)^{\otimes n+1},\nonumber
\end{align}
where \cref{eq-692323} is because $U$ and $U'$ are Haar random unitary.
\end{proof}

\begin{fact}\label{fact-5121847}
$\E_U[C_{U}]$ is invariant under simultaneous permutation, i.e., for any $\pi\in\mathfrak{S}_{n+1}$, let $P^{\textup{odd}}(\pi)$ be the action of $\pi$ on $\mathcal{H}_1\otimes \mathcal{H}_3\otimes \cdots \mathcal{H}_{2n+1}\cong(\mathbb{C}^d)^{\otimes n+1}$ and $P^{\textup{even}}(\pi)$ be that on $\mathcal{H}_2\otimes \mathcal{H}_4\otimes \cdots \mathcal{H}_{2n+2}\cong(\mathbb{C}^d)^{\otimes n+1}$, then
\[P^{\textup{even}}(\pi)\otimes P^{\textup{odd}}(\pi)\E_U[C_{U}]=\E_U[C_{U}]P^{\textup{even}}(\pi)\otimes P^{\textup{odd}}(\pi)=\E_U[C_{U}].\]
\end{fact}
\begin{proof}[Proof of \cref{fact-5121847}]
\[P^{\textup{even}}(\pi)\otimes P^{\textup{odd}}(\pi)\E_U[C_{U}]=\E_U\!\left[P^{\textup{even}}(\pi)\otimes P^{\textup{odd}}(\pi)\kettbbra{U}{U}^{\otimes n+1}\right]=\E_U\!\left[\kettbbra{U}{U}^{\otimes n+1}\right]=\E_U[C_{U}].\]
The other direction is similar.
\end{proof}

Note that we have the following decomposition of the space 
\begin{align}
(\mathcal{H}_1\otimes \mathcal{H}_3 \otimes \cdots\otimes \mathcal{H}_{2n+1}) \otimes (\mathcal{H}_2\otimes \mathcal{H}_4\otimes \cdots \otimes \mathcal{H}_{2n+2})  \cong &\,\,(\mathbb{C}^d)^{\otimes n+1}\otimes (\mathbb{C}^d)^{\otimes n+1}\nonumber\\
\stackrel{\mathbb{U}_d\times \mathbb{U}_d\times \mathfrak{S}_{n+1}\times \mathfrak{S}_{n+1}}{\cong} & \bigoplus_{\lambda,\mu\vdash_d n+1} \mathcal{Q}_\lambda\otimes \mathcal{Q}_\mu\otimes \mathcal{P}_\lambda\otimes \mathcal{P}_\mu.\nonumber
\end{align}

Therefore, by \cref{fact-5121815} and Schur's lemma, we have
\[\E_U[C_{U}]=\bigoplus_{\lambda,\mu\vdash_d n+1} I_{\mathcal{Q}_\lambda}\otimes I_{\mathcal{Q}_\mu} \otimes M_{\lambda,\mu},\]
for some $M_{\lambda,\mu}\in\mathcal{L}(\mathcal{P}_\lambda\otimes \mathcal{P}_\mu)$. 

On the other hand, letting $P_{\avg}=\frac{1}{(n+1)!}\sum_{\pi\in\mathfrak{S}_{n+1}} P^{\textup{even}}(\pi)\otimes P^{\textup{odd}}(\pi)$, we have
\begin{align}
P_{\avg} &=\bigoplus_{\lambda,\mu\vdash_d n+1} I_{\mathcal{Q}_\lambda}\otimes I_{\mathcal{Q}_\mu}\otimes\frac{1}{(n+1)!}\sum_{\pi\in\mathfrak{S}_{n+1}} P_\lambda(\pi)\otimes P_\mu(\pi)\\
&=\bigoplus_{\lambda,\mu\vdash_d n+1} I_{\mathcal{Q}_\lambda}\otimes I_{\mathcal{Q}_\mu} \otimes \frac{1}{(n+1)!}\sum_{\pi\in\mathfrak{S}_{n+1}}P_\lambda(\pi)\otimes P^*_\mu(\pi)\label{eq-5121832}\\
&=\bigoplus_{\lambda\vdash_d n+1} I_{\mathcal{Q}_\lambda}\otimes I_{\mathcal{Q}_\lambda} \otimes \frac{1}{\dim(\mathcal{P}_\lambda)} \kettbbra{I_{\mathcal{P}_\lambda}}{I_{\mathcal{P}_\lambda}},\label{eq-5121835}
\end{align}
where $P_{\lambda}(\pi)$ is the action of $\pi$ on $\mathcal{P}_\lambda$, \cref{eq-5121832} is because $P_{\mu}(\pi)$ is a real-valued matrix on the Young basis, \cref{eq-5121835} is because the subspace in $\mathcal{P}_\lambda\otimes \mathcal{P}_\mu$ that is invariant under $P_{\lambda}(\pi)\otimes P_{\mu}^*(\pi)$ is $\textup{span}(\kett{I_{\mathcal{P}_\lambda}})$ when $\lambda=\mu$ and the trivial subspace $\{0\}$ when $\lambda\neq \mu$.

Then, \cref{fact-5121847} implies that
$P_{\avg} \E_U[C_{U}]P_{\avg}=\E_U[C_{U}]$, which means
\[\E_U[C_{U}]=\bigoplus_{\lambda\vdash_d n+1} z_{\lambda}\cdot  I_{\mathcal{Q}_\lambda}\otimes I_{\mathcal{Q}_\lambda}\otimes  \kettbbra{I_{\mathcal{P}_\lambda}}{I_{\mathcal{P}_\lambda}},\]
for some $z_{\lambda}\in\mathbb{C}$.

Moreover, since $\tr_{2,4,\ldots,2n+2}(C_{U})=I^{\otimes n+1}$, we have
\begin{align}
I^{\otimes n+1}&= \bigoplus_{\lambda\vdash_d n+1} z_\lambda \cdot \tr_{\mathcal{Q}_\lambda}(I_{\mathcal{Q}_\lambda})\cdot I_{\mathcal{Q}_\lambda}\otimes \tr_{\mathcal{P}_\lambda}(\kettbbra{I_{\mathcal{P}_\lambda}}{I_{\mathcal{P}_\lambda}})\\
&=\bigoplus_{\lambda\vdash_d n+1}z_\lambda\cdot \dim(\mathcal{Q}_\lambda)\cdot I_{\mathcal{Q}_\lambda} \otimes I_{\mathcal{P}_\lambda},
\end{align}
which means $z_\lambda=\frac{1}{\dim(\mathcal{Q}_\lambda)}$. Therefore,
\[\E_U[C_{U}]=\bigoplus_{\lambda \vdash_d n+1}\frac{1}{\dim(\mathcal{Q}_\lambda)}\cdot I_{\mathcal{Q}_\lambda}\otimes I_{\mathcal{Q}_\lambda} \otimes \kettbbra{I_{\mathcal{P}_\lambda}}{I_{\mathcal{P}_\lambda}}.\]
\end{proof}

\begin{lemma}\label{fact-5122103}
For any positive semidefinite matrix $M$ and vector $\ket{\psi}$ such that $\ket{\psi}\in\supp(M)$, we have
\[M\sqsupseteq \ketbra{\psi}{\psi}\Longleftrightarrow 1\geq \bra{\psi}M^{-1}\ket{\psi},\]
where $M^{-1}$ is the pseudo-inverse of $M$.
\end{lemma}
\begin{proof}
\[M\sqsupseteq \ketbra{\psi}{\psi}\Longleftrightarrow I_{\supp(M)} \sqsupseteq M^{-1/2}\ketbra{\psi}{\psi} M^{-1/2},\]
where $M^{-1/2}$ is the pseudo-inverse of $M^{1/2}$. Then
\[I_{\supp(M)} \sqsupseteq M^{-1/2}\ketbra{\psi}{\psi} M^{-1/2}\Longleftrightarrow 1\geq \tr(M^{-1/2}\ketbra{\psi}{\psi} M^{-1/2}) = \bra{\psi}M^{-1}\ket{\psi}.\]
\end{proof}

\section*{Acknowledgment}

We thank John Wright and Ewin Tang for pointing out an error in our previous result on the hardness of unitary controlization and for directing us to the relevant reference~\cite{sheridan2009approximating}. We also thank Qisheng Wang for helpful discussions. Finally, we thank the anonymous reviewers for their valuable comments and suggestions.

The work of Zhicheng Zhang was supported in part by the Australian Research Council under Grant DP250102952.

\bibliographystyle{alpha}
\bibliography{main}

\end{document}